\documentclass{article}
\usepackage{helvet}
\usepackage{courier}
\usepackage[T1]{fontenc}
\usepackage[a4paper]{geometry}
\usepackage[active]{srcltx}
\usepackage{color}
\usepackage{verbatim}
\usepackage{mathrsfs}
\usepackage{mathtools}
\usepackage{dsfont}
\usepackage{amsmath}
\usepackage{amsthm}
\usepackage{amssymb}
\usepackage{makeidx}
\makeindex
\usepackage[pdfusetitle,
 bookmarks=true,bookmarksnumbered=true,bookmarksopen=true,bookmarksopenlevel=2,
 breaklinks=false,pdfborder={0 0 1},backref=page,colorlinks=false]
 {hyperref}
\hypersetup{
 unicode=false, linkcolor=black, citecolor=black, urlcolor=blue, filecolor=blue, pdfpagelayout=OneColumn, pdfnewwindow=true, pdfstartview=XYZ, plainpages=false}

\makeatletter
\numberwithin{equation}{section}
\theoremstyle{plain}
\newtheorem{thm}{\protect\theoremname}[section]
\theoremstyle{remark}
\newtheorem{claim}[thm]{\protect\claimname}
\theoremstyle{definition}
\newtheorem{defn}[thm]{\protect\definitionname}
\theoremstyle{plain}
\newtheorem{prop}[thm]{\protect\propositionname}
\theoremstyle{plain}
\newtheorem{lyxalgorithm}[thm]{\protect\algorithmname}
\theoremstyle{plain}
\newtheorem{assumption}[thm]{\protect\assumptionname}
\theoremstyle{plain}
\newtheorem{lem}[thm]{\protect\lemmaname}
\theoremstyle{plain}
\newtheorem{cor}[thm]{\protect\corollaryname}
\theoremstyle{remark}
\newtheorem{rem}[thm]{\protect\remarkname}
\theoremstyle{plain}
\newtheorem{question}[thm]{\protect\questionname}

\usepackage{color}
\usepackage{slashed}

 \let\myTOC\tableofcontents
 \renewcommand\tableofcontents{%
   \pdfbookmark[1]{Contents}{}
   \myTOC
   \cleardoublepage
   \pagenumbering{arabic} }

\global\long\def\foreignlanguage#1#2{#2}%
\global\long\def\selectlanguage#1{}%

\allowdisplaybreaks

\DeclareMathOperator\supp{\mathrm{supp}}
\DeclareMathOperator\Tr{\mathrm{Tr}}

\makeatother

\providecommand{\algorithmname}{Algorithm}
\providecommand{\assumptionname}{Assumption}
\providecommand{\claimname}{Claim}
\providecommand{\corollaryname}{Corollary}
\providecommand{\definitionname}{Definition}
\providecommand{\lemmaname}{Lemma}
\providecommand{\propositionname}{Proposition}
\providecommand{\questionname}{Question}
\providecommand{\remarkname}{Remark}
\providecommand{\theoremname}{Theorem}

\begin{document}
\title{A skepticism on the concept of quantum state related to quantum field
theory on curved spacetime}
\author{Hideyasu Yamashita\\
{\small Division of Liberal Arts and Sciences, Aichi-Gakuin University}\\
\texttt{\small yamasita@dpc.aichi-gakuin.ac.jp}
}

\date{\today}

\maketitle
\newcommand{\dcolor}{\definecolor{note_fontcolor}{rgb}{0.1, 0.0, 0.8}}
\definecolor{HYnote_fontcolor}{rgb}{0.2, 0.0, 0.2}
\definecolor{hycolor}{rgb}{0.3, 0.0, 0.3}
\newcommand{\hyc}{\color{hycolor}}
\newenvironment{HYnote}
 {\textcolor{note_fontcolor}\bgroup\ignorespaces}
  {\ignorespacesafterend\egroup} 

\newenvironment{trivenv}
  {\bgroup\ignorespaces}
  {\ignorespacesafterend\egroup}

\newcommand{\displabel}[1]{}

\newcommand{\hidable}[3]{#2}
\newcommand{\hidea}[1]{{#1}}
\newcommand{\hideb}[1]{{#1}}
\newcommand{\hidec}[1]{{#1}}
\newcommand{\hidep}[1]{{#1}}
\renewcommand{\hidec}[1]{}
\renewcommand{\hidep}[1]{}

\newcommand{\thlab}[1]{{\tt [#1]}}

\newcommand{\black}{\color{black}}

\global\long\def\N{\mathbb{N}}%
\global\long\def\C{\mathbb{C}}%
\global\long\def\Z{\mathbb{Z}}%
 
\global\long\def\R{\mathbb{R}}%
 
\global\long\def\im{\mathrm{i}}%

\global\long\def\di{\partial}%
 
\global\long\def\d{{\rm d}}%

\global\long\def\ol#1{\overline{#1}}%
\global\long\def\ul#1{\underline{#1}}%
\global\long\def\ob#1{\overbrace{#1}}%

\global\long\def\ov#1{\overline{#1}}%

\global\long\def\then{\Rightarrow}%
 
\global\long\def\Then{\Longrightarrow}%

\global\long\def\N{\mathbb{N}}%
\global\long\def\C{\mathbb{C}}%
\global\long\def\Z{\mathbb{Z}}%
 
\global\long\def\R{\mathbb{R}}%
 
\global\long\def\im{\mathrm{i}}%

\global\long\def\di{\partial}%
 
\global\long\def\d{{\rm d}}%

\global\long\def\ol#1{\overline{#1}}%
\global\long\def\ul#1{\underline{#1}}%
\global\long\def\ob#1{\overbrace{#1}}%

\global\long\def\ov#1{\overline{#1}}%

\global\long\def\then{\Rightarrow}%
 
\global\long\def\Then{\Longrightarrow}%

\global\long\def\cA{\mathcal{A}}%
\global\long\def\cB{\mathcal{B}}%
 
\global\long\def\cC{\mathcal{C}}%
 
\global\long\def\cD{\mathcal{D}}%
\global\long\def\cE{\mathcal{E}}%
 
\global\long\def\cF{\mathcal{F}}%
 
\global\long\def\cG{{\cal G}}%
 
\global\long\def\cH{\mathcal{H}}%
 
\global\long\def\cI{\mathcal{I}}%
 
\global\long\def\cJ{\mathcal{J}}%
\global\long\def\cK{\mathcal{K}}%
 
\global\long\def\cL{\mathcal{L}}%
 
\global\long\def\cM{\mathcal{M}}%
 
\global\long\def\cN{\mathcal{N}}%
 
\global\long\def\cO{\mathcal{O}}%
 
\global\long\def\cP{\mathcal{P}}%
 
\global\long\def\cQ{\mathcal{Q}}%
 
\global\long\def\cR{\mathcal{R}}%
 
\global\long\def\cS{\mathcal{S}}%
 
\global\long\def\cT{\mathcal{T}}%
 
\global\long\def\cU{\mathcal{U}}%
 
\global\long\def\cV{\mathcal{V}}%
 
\global\long\def\cW{\mathcal{W}}%
\global\long\def\cX{\mathcal{X}}%
 
\global\long\def\cY{\mathcal{Y}}%
 
\global\long\def\cZ{\mathcal{Z}}%

\global\long\def\scA{\mathscr{A}}%
\global\long\def\scB{\mathscr{B}}%
\global\long\def\scC{\mathscr{C}}%
\global\long\def\scD{\mathscr{D}}%
 
\global\long\def\scE{\mathscr{E}}%
 
\global\long\def\scF{\mathscr{F}}%
 
\global\long\def\scG{\mathscr{G}}%
 
\global\long\def\scH{\mathscr{H}}%
 
\global\long\def\scI{\mathscr{I}}%
 
\global\long\def\scJ{\mathscr{J}}%
 
\global\long\def\scK{\mathscr{K}}%
 
\global\long\def\scL{\mathscr{L}}%
 
\global\long\def\scM{\mathscr{M}}%
 
\global\long\def\scN{\mathscr{N}}%
 
\global\long\def\scO{\mathscr{O}}%
 
\global\long\def\scP{\mathscr{P}}%
 
\global\long\def\scR{\mathscr{R}}%
\global\long\def\scS{\mathscr{S}}%
 
\global\long\def\scT{\mathscr{T}}%
 
\global\long\def\scU{\mathscr{U}}%
 
\global\long\def\scW{\mathscr{W}}%
\global\long\def\scZ{\mathscr{Z}}%

\global\long\def\bbA{\mathbb{A}}%
 
\global\long\def\bbB{\mathbb{B}}%
 
\global\long\def\bbD{\mathbb{D}}%
 
\global\long\def\bbE{\mathbb{E}}%
 
\global\long\def\bbF{\mathbb{F}}%
 
\global\long\def\bbG{\mathbb{G}}%
 
\global\long\def\bbI{\mathbb{I}}%
 
\global\long\def\bbJ{\mathbb{J}}%
 
\global\long\def\bbK{\mathbb{K}}%
 
\global\long\def\bbL{\mathbb{L}}%
 
\global\long\def\bbM{\mathbb{M}}%
 
\global\long\def\bbP{\mathbb{P}}%
 
\global\long\def\bbQ{\mathbb{Q}}%
 
\global\long\def\bbT{\mathbb{T}}%
 
\global\long\def\bbU{\mathbb{U}}%
 
\global\long\def\bbX{\mathbb{X}}%
 
\global\long\def\bbY{\mathbb{Y}}%
\global\long\def\bbW{\mathbb{W}}%

\global\long\def\bbOne{1\kern-0.7ex  1}%

\renewcommand{\bbOne}{\mathbbm{1}}

\global\long\def\bB{\mathbf{B}}%
 
\global\long\def\bG{\mathbf{G}}%
 
\global\long\def\bH{\mathbf{H}}%
\global\long\def\bS{\boldsymbol{S}}%
 
\global\long\def\bT{\mathbf{T}}%
 
\global\long\def\bX{\mathbf{X}}%
\global\long\def\bY{\mathbf{Y}}%
\global\long\def\bW{\mathbf{W}}%
 
\global\long\def\boT{\boldsymbol{T}}%

\global\long\def\fraka{\mathfrak{a}}%
 
\global\long\def\frakb{\mathfrak{b}}%
 
\global\long\def\frakc{\mathfrak{c}}%
 
\global\long\def\frake{\mathfrak{e}}%
 
\global\long\def\frakf{\mathfrak{f}}%
 
\global\long\def\fg{\mathfrak{g}}%
 
\global\long\def\frakh{\mathfrak{h}}%
 
\global\long\def\fraki{\mathfrak{i}}%
\global\long\def\frakk{\mathfrak{k}}%
 
\global\long\def\frakl{\mathfrak{l}}%
 
\global\long\def\frakm{\mathfrak{m}}%
 
\global\long\def\frakn{\mathfrak{n}}%
 
\global\long\def\frako{\mathfrak{o}}%
 
\global\long\def\frakp{\mathfrak{p}}%
 
\global\long\def\frakq{\mathfrak{q}}%
 
\global\long\def\fraks{\mathfrak{s}}%
 
\global\long\def\fs{\mathfrak{s}}%
 
\global\long\def\fraku{\mathfrak{u}}%
\global\long\def\frakz{\mathfrak{z}}%

\global\long\def\fA{\mathfrak{A}}%
 
\global\long\def\fB{\mathfrak{B}}%
 
\global\long\def\fC{\mathfrak{C}}%
 
\global\long\def\fD{\mathfrak{D}}%
 
\global\long\def\fF{\mathfrak{F}}%
 
\global\long\def\fG{\mathfrak{G}}%
 
\global\long\def\fK{\mathfrak{K}}%
 
\global\long\def\fL{\mathfrak{L}}%
 
\global\long\def\fM{\mathfrak{M}}%
 
\global\long\def\fP{\mathfrak{P}}%
 
\global\long\def\fR{\mathfrak{R}}%
 
\global\long\def\fS{\mathfrak{S}}%
\global\long\def\fT{\mathfrak{T}}%
 
\global\long\def\fU{\mathfrak{U}}%
\global\long\def\fV{\mathfrak{V}}%
 
\global\long\def\fW{\mathfrak{W}}%
\global\long\def\fX{\mathfrak{X}}%
\global\long\def\fZ{\mathfrak{Z}}%

\global\long\def\ssS{\mathsf{S}}%
\global\long\def\ssT{\mathsf{T}}%
 
\global\long\def\ssW{\mathsf{W}}%

\global\long\def\rM{\mathrm{M}}%
\global\long\def\prj{\mathfrak{P}}%

{} 
\global\long\def\sy#1{{\color{blue}#1}}%

\global\long\def\magenta#1{{\color{magenta}#1}}%

\global\long\def\symb#1{#1}%

\global\long\def\emhrb#1{\text{{\color{red}{\huge {\bf #1}}}}}%

\newcommand{\symbi}[1]{\index{$ #1$}{\color{red}#1}} 

{} 
\global\long\def\SYM#1#2{#1}%

\renewcommand{\SYM}[2]{\symb{#1}}

\newcommand{\usuji}{\color[rgb]{0.7,0.4,0.4}} \newcommand{\usu}{\color[rgb]{0.5,0.2,0.1}}
\newenvironment{Usuji} {\begin{trivlist}   \item \usuji }  {\end{trivlist}}
\newenvironment{Usu} {\begin{trivlist}   \item \usu }  {\end{trivlist}} 

\newcommand{\term}[1]{\textcolor[rgb]{0, 0, 1}{\bf #1}}
\newcommand{\termi}[1]{{\bf #1}}

\newcommand{\slim}{\mathop{\mbox{s-lim}}} %

\newcommand{\wlim}{\mathop{\mbox{w-lim}}}

\newcommand{\limsub}{\mathop{\mbox{\rm lim-sub}}}

\global\long\def\bboxplus{\boxplus}%

\renewcommand{\bboxplus}{\mathop{\raisebox{-0.8ex}{\text{\begin{trivenv}\LARGE{}$\boxplus$\end{trivenv}}}}}

\global\long\def\shuff{\sqcup\kern-0.3ex  \sqcup}%

\renewcommand{\shuff}{\shuffle}

\global\long\def\upha{\upharpoonright}%

\global\long\def\ket#1{|#1\rangle}%
 
\global\long\def\bra#1{\langle#1|}%

{} 
\global\long\def\lll{\vert\kern-0.25ex  \vert\kern-0.25ex  \vert}%
 \renewcommand{\lll}{{\vert\kern-0.25ex  \vert\kern-0.25ex  \vert}}

\global\long\def\biglll{\big\vert\kern-0.25ex  \big\vert\kern-0.25ex  \big\vert\kern-0.25ex  }%
 
\global\long\def\Biglll{\Big\vert\kern-0.25ex  \Big\vert\kern-0.25ex  \Big\vert}%

\newcommand{\iiia}[1]{{\left\vert\kern-0.25ex\left\vert\kern-0.25ex\left\vert #1
  \right\vert\kern-0.25ex\right\vert\kern-0.25ex\right\vert}}

\global\long\def\iii#1{\iiia{#1}}%

\global\long\def\Upa{\Uparrow}%
 
\global\long\def\Nor{\Uparrow}%

\newcommand{\vertt}{\kern-0.6ex\vert}
\renewcommand{\Nor}{[\kern-0.16ex ]}

\global\long\def\Prob{\mathbb{P}}%
\global\long\def\Var{\mathrm{Var}}%
\global\long\def\Cov{\mathrm{Cov}}%
\global\long\def\Ex{\mathbb{E}}%
{} 
\global\long\def\Ae{{\rm a.e.}}%
 
\global\long\def\samples{\bOm}%


\global\long\def\bOne{{\bf 1}}%

\global\long\def\Ten{\bullet}%
{} %

\global\long\def\TT{\intercal}%
 \renewcommand{\TT}{\mathsf{T}}

\global\long\def\trit{\vartriangle\!\! t}%

\global\long\def\bProj{\mathfrak{P}}%
\global\long\def\bE{\mathbf{E}}%

\global\long\def\grad{\mathrm{grad}}%
 
\global\long\def\Hom{\mathrm{Hom}}%

\global\long\def\Inv{{\rm Inv}}%

\global\long\def\Lie{{\rm Lie}}%
 
\global\long\def\leng{\text{{\rm leng}}}%
 
\global\long\def\meas{\text{{\rm meas}}}%

\global\long\def\GreenOp{\mathsf{E}}%
\global\long\def\Sol{\mathsf{Sol}}%
\global\long\def\GL{{\rm GL}}%

\global\long\def\Pow{\mathsf{P}}%

\global\long\def\p{\mathbf{p}}%
 
\global\long\def\q{\mathbf{q}}%

\global\long\def\rF{\mathrm{F}}%
 
\global\long\def\rE{\mathrm{E}}%

\global\long\def\sfS{\mathsf{S}}%
\global\long\def\sfQ{\mathsf{Q}}%

\global\long\def\spec{{\rm spec}}%
 
\global\long\def\Sp{{\rm Sp}}%
\global\long\def\sfD{\mathsf{D}}%
 
\global\long\def\sperp{\angle}%
{} 

{} %

\global\long\def\WICK#1{{\rm w}\{#1\}}%
 \renewcommand{\WICK}[1]{{:}#1{:}}

\global\long\def\AProj{{\rm AProj}}%
\global\long\def\Ad{{\rm Ad}}%
\global\long\def\ad{{\rm ad}}%
\global\long\def\Borel{{\rm Borel}}%
\global\long\def\area{{\bf S}}%
\global\long\def\bbS{\mathbb{S}}%
\global\long\def\bbOne{\mathds{1}}%
\global\long\def\Bdd{\mathscr{B}}%
\global\long\def\Borel{{\rm Borel}}%
\global\long\def\bP{{\bf P}}%

\global\long\def\Cl{{\rm C}\ell}%
\global\long\def\cconj{\blacklozenge}%
\global\long\def\cpt{{\rm c}}%
\global\long\def\cc{{\bf c}}%
\global\long\def\curve{\mathsf{C}}%

\global\long\def\CCRW{\mathcal{CCR}}%
\global\long\def\ccr{{\rm Ccr}}%
\global\long\def\CCR{{\rm CCR}}%
\global\long\def\CAR{{\rm CAR}}%
\global\long\def\ComplexS{J}%

\global\long\def\diam{\text{{\rm diam}}}%
\global\long\def\dom{\mathrm{dom}}%
\global\long\def\End{{\rm End}}%
\global\long\def\ex{{\rm ex}}%

\global\long\def\Fin{{\rm Fin}}%

\global\long\def\DiracOpm{{\bf \mathsf{D}}}%
\global\long\def\Dconj{\mathfrak{d}}%
\global\long\def\DSpinors{\mathscr{C}}%
\global\long\def\even{{\rm even}}%

\global\long\def\GreenOp{\mathsf{S}}%
\global\long\def\hol{{\rm hol}}%

\global\long\def\Id{{\rm Id}}%
 
\global\long\def\id{{\rm id}}%

\newcommand{\Kahler}{K{\"a}hler}

\global\long\def\LBundle{\cL}%
 
\global\long\def\Leb{\text{{\rm Leb}}}%

\global\long\def\Manifold{\mathscr{X}}%
 
\global\long\def\Mat{{\rm Mat}}%

\global\long\def\nablaslash{\slashed{\nabla}}%

\global\long\def\ON{{\rm ON}}%
 
\global\long\def\OCpl{{\rm OC}}%
\global\long\def\Odd{{\rm Odd}}%

\global\long\def\Proj{{\rm Proj}}%
\global\long\def\cauindep{\pitchfork}%

\global\long\def\qdev{{\rm qd}}%
\global\long\def\Paths{\scE}%
 
\global\long\def\Pin{{\rm Pin}}%
 
\global\long\def\POVM{\mathsf{M}}%
\global\long\def\STime{\mathcal{M}}%

\global\long\def\quantity{\mathsf{q}}%
\global\long\def\Quantities{\mathscr{Q}}%
\global\long\def\ran{\mathrm{ran}}%

\global\long\def\rD{{\rm D}}%
\global\long\def\re{{\bf r}}%
\global\long\def\RDSPinors{\scR}%

\global\long\def\subset{\subseteq}%
\global\long\def\sgn{{\rm sgn}}%
\global\long\def\Span{{\rm span}}%
\global\long\def\Spin{{\rm Spin}}%
\global\long\def\spin{\mathfrak{spin}}%
\global\long\def\Sol{\mathsf{Sol}_{{\rm sc}}}%
\global\long\def\vac{{\rm vac}}%

\global\long\def\weyl{\mathsf{W}}%
\global\long\def\x{\mathbf{x}}%
 
\global\long\def\y{\mathbf{y}}%

\def\foreignlanguage#1#2{#2}


\let\ruleorig=\rule
\renewcommand{\rule}{\noindent\ruleorig}

\global\long\def\labelenumi{(\arabic{enumi})}%

\newcommand{\reff}{\ref}

\begin{abstract}
Some skeptical arguments on the physical reality of quantum states
are given. First, I argue that the algebraic formalism of quantum
field theory in curved spacetime (algebraic QFTCS, AQFTCS) leads to
such a skepticism.%
{} Of course we have the purely mathematical notion of states on a $C^{*}$-algebra
$\mathfrak{A}$, but usually in non-relativistic quantum mechanics
and quantum field theory in Minkowski spacetime (QFTM), not all of
them are considered to be physically real; Some of them are physically
real (or realizable) states, but others are non-physical ``fictional''
states. Only the states which can be expressed as a density matrix
on a fixed ``physical Hilbert space'' (the GNS representation space
of $\mathfrak{A}$ w.r.t.~the vacuum) are viewed to be physically
real. On the other hand, in QFTCS, there is no distinguished physical
Hilbert space; no distinguished vacuum state. Thus we cannot distinguish
physically real states from fictional states. The second part of my
argument is a counterargument to what I call ``pragmatic realism
on quantum states'', which insists as follows: ``We are permitted
to regard a quantum state as a physical reality, because the concept
of quantum state is indispensable in quantum physics.'' I argue that
the concept of quantum state is indeed dispensable in non-relativistic
QM, and hence this pragmatic realist thesis is vacuous there. I give
a conjecture that it is also dispensable in QFTM and QFTCS, and some
preliminary considerations on it.
\end{abstract}

\section{Introduction}

There are several skeptical stances regarding the physical reality
of quantum states. Two major examples of such stances will be the
following: \emph{relational quantum mechanics} (RQM) (Rovelli \cite{Rov1996})
and \emph{quantum Bayesianism} (QBism) (e.g., \cite{Fuc2002,DFPS2021}).

Although Rovelli \cite{Rov1996} would not go deeply into the ontological
arguments, the following theses of him indicate that he takes an anti-realist
stance on quantum states:

\begin{quote}
The notion rejected here is the notion of absolute, or observer-independent,
state of a system; equivalently, the notion of observer-independent
values of physical quantities. \cite[introduction]{Rov1996}
\end{quote}
\begin{quote}
If different observers give different descriptions of the state of
the same system, this means that the notion of state is observer dependent.
\cite[Section 4]{Rov1996}
\end{quote}
While also QBism emphasizes the observer-dependence of quantum states,
it is known to be more radical, extremist and disputatious (e.g.,
see the bombast of Fuchs \cite[Sec.9]{Fuc2002}: ``QUANTUM STATES
DO NOT EXIST'').

In this paper, I will give some skeptical arguments on the physical
reality of quantum states. Although I am more or less sympathetic
to both RQM and QBism, I avoid any direct evaluation of them in this
article, so that my arguments are logically independent from the the
debates on RQM and QBism. However it should be noted that some motivations
and ideas in this article are inspired by them. On the other hand,
also note that there is a basic difference between my stance and theirs
in the following point: I think that until the clear answer to the
question ``what is an observer?'' is known, it is desirable to avoid
referring to an observer as well as possible, in quantum physics;
Instead I would like to emphasize the ``scope-dependence'' rather
than the observer-dependence. While I shall not define the precise
meaning of ``scope'' in this paper, I use this term as an analogue
to the term ``scope'' commonly used in logic and linguistics.

To formalize the notion of scope-dependence, I suggest that the mathematical
notion of causal net of the algebras \cite{BR87,Haa96,Araki99} of
local observables serves as a common algebraic framework for the three
types of theories: non-relativistic quantum mechanics, quantum field
theory in Minkowski spacetime (QFTM) and quantum field theory in curved
spacetime (QFTCS) in Section \ref{sec:The-causal-net}.%

In Sections \ref{sec:Global-state} and \ref{sec:Local-state}, I
argue that the algebraic formalism of quantum physics, especially
that of algebraic QFTCS (AQFTCS, \cite{Wal1994,BFV2003,BF2009,HW2015,KM2015,FR2016})
leads to a skepticism on the physical reality of quantum states.%
{} Of course we have the purely mathematical notion of states on a $C^{*}$-algebra
$\mathfrak{A}$, but usually in non-relativistic quantum mechanics
and QFTM, not all of them are considered to be physically real; Some
of them are physically real (or realizable) states, but others are
non-physical ``fictional'' states. Only the states which can be
expressed as a density matrix on a fixed ``physical Hilbert space''
(the GNS representation space of $\mathfrak{A}$ w.r.t.~the vacuum)
are viewed to be physically real. On the other hand, in QFTCS, there
is no distinguished physical Hilbert space; no distinguished vacuum
state. Thus we cannot distinguish physically real states from fictional
states.

However it might be logically possible for a ``quantum state realist''
to insist as follows: ``There should be the physically real states
distinguished from the fictional states, even if we cannot know how
to distinguish them''. This view may be partially supported by what
I call ``pragmatic realism on quantum states''.

{} First consider the analogous problem on the physical reality of vector
potentials. In the classical theory of electromagnetic fields, the
notion of vector potential is a very useful tool, but considered to
be dispensable. On the other hand, in the quantum theory concerning
electromagnetic fields, that notion appears to become almost indispensable
(e.g., to explain the Aharonov\textendash Bohm effect).

The argument of pragmatic realists can be expressed as follows.
\begin{claim}
\textbf{\label{claim:Pragmatic-realism-for-vector}Pragmatic realism
for vector potentials}: We are permitted to regard a vector potential
as a physical \emph{reality}, because the concept of vector potential
is \emph{indispensable} for us in quantum physics.
\end{claim}

Although probably the implication ``indispensable $\Rightarrow$
real'' will not be justified universally, Claim \ref{claim:Pragmatic-realism-for-vector}
is fairly persuasive by its simple logic. However, if the premise
``vector potentials are indispensable'' is disproved, this argument
immediately becomes vacuous. \footnote{I conjecture that the ``loop-wise'' observables such as the Wilson
loops, which are related with the geometric notion of \emph{holonomy},
can make vector potentials dispensable.}

Similarly, a pragmatic realist for quantum states would say as follows:
\begin{claim}
\textbf{\label{claim:Pragmatic-realism-for-states}Pragmatic realism
for quantum states}: We are permitted to regard a quantum state as
a physical \emph{reality}, because the concept of quantum state is
\emph{indispensable} in quantum physics.
\end{claim}

I shall argue that the concept of quantum state is indeed dispensable
in non-relativistic QM, and hence this thesis of pragmatic realists
is vacuous there, in Section \ref{sec:NRQM-withoutStates}. Furthermore
I give a conjecture that it is also dispensable in QFTM and QFTCS,
and some preliminary considerations on it in Section \ref{sec:QFTM-and-QFTCS}.
Here I do not discuss the (in)dispensability of states in the purely
mathematical theory of $C^{*}$- and $W^{*}$-algebras. When I refer
to an ``(in)dispensability'' in this paper, it refers to a ``\emph{physical
(in)dispensability}'', or more precisely an ``\emph{empirical (in)dispensability}'',
that is, whether the concept of state is (in)dispensable to describe
the physical/empirical laws which can be verified or falsified experimentally.
I expect that also the concept of observer is empirically dispensable,
but that is not the main topic in this paper.

My skepticism on the concept of state is one of motivations of the
papers \cite{Yam2025,Yam2026b}. The detailed explanation of this
motivation would be somewhat lengthy, so I have written it here as
a separate paper.

\section{The causal net of scope-dependent algebras of observables}

\label{sec:The-causal-net}

In this section, let us consider a common algebraic framework for
the three types of theories: non-relativistic quantum mechanics, QFTM
and QFTCS. This framework shall be called causal nets (Definition
\ref{def:causalNet} below). Although the notion of causal net has
been employed mainly in algebraic QFT in Minkowski spacetime (AQFTM)
\cite{Haa96,Araki99}, it is used also in non-relativistic quantum
physics (e.g., \cite[Sec.2.6]{BR87}).

Let $\fA$ be a unital {*}-algebra. A \termi{state} on $\fA$ is
a linear functional $\omega:\fA\to\C$ such that
\[
\omega(\bbOne)=1,\qquad\omega(x^{*}x)\ge0,\quad\text{for all }x\in\fA.
\]
Since a generic {*}-algebra is not given any topology, probably this
is not suitable to develop a general abstract theory. In many practical
examples, $\fA$ is generated by the elements $\{e_{\lambda}|\lambda\in\Lambda\}$,
where $\Lambda$ is a topological space with some other structures.
In this case, we can use the topology (and other structures) of $\Lambda$,
rather than that of $\fA$ itself. Wightman's formalism of QFTM \cite{SW64}
can be understood as an example of the theory of this type.

If $\fA$ is a $C^{*}$-algebra, ``a state on $\fA$'' refers to
a state on the {*}-algebra $\fA$ which is continuous w.r.t.~the
norm $\|\cdot\|$. Furthermore, if $\fA$ is a $W^{*}$-algebra, usually
only the \emph{normal} states (defined below) on $\fA$ are considered,
but ``a state on $\fA$'' refers to a (possibly non-normal) state
on $\fA$ as a $C^{*}$-algebra. 

Assume that $\fA$ can be interpreted as an algebra of observables.
Roughly speaking, if a state $\omega$ on a $W^{*}$-algebra $\fA$
is physically realizable, $\omega$ should be normal. Also the converse
is considered to hold if $\fA$ is a factor of type I (see below).
However, it is not clear whether the converse holds in general.

If $\fA$ is a $C^{*}$-algebra which is {*}-isomorphic with a von
Neumann algebra, it is called a \termi{$W^{*}$-algebra}. Alternatively,
we can define it by \termi{Kadison's characterization of $W^{*}$-algebras},
stated as follows.

Let $\fA_{{\rm h}}$ denote the set of all self-adjoint elements of
$\fA$.
\begin{defn}
\cite[Definition 3.13, p.137]{Tak2002} %
A $C^{*}$-algebra $\fA$ is said to be \termi{monotone closed} if
every bounded increasing net in $\fA_{{\rm h}}$ has the least upper
bound in $\fA_{{\rm h}}$. A positive linear functional $\omega$
on a monotone closed $C^{*}$-algebra $A$ is called \termi{normal}
if $\omega(\sup x_{i})=\sup\omega(x_{i})$ for every bounded increasing
net $\{x_{i}\}$ in $\fA_{{\rm h}}$, where $\sup x_{i}$ means the
least upper bound of $\{{x_{i}}\}$ in $\fA_{{\rm h}}$.
\end{defn}

\begin{thm}[{Kadison. See \cite[Thm.3.16, p.138]{Tak2002}}]
{} Let $\mathfrak{A}$ be a $C^{*}$-algebra. Then $\fA$ is a $W^{*}$-algebra
if and only if
\begin{enumerate}
\item $\fA$ is monotone closed, and
\item for any positive nonzero element $A\in\mathfrak{A}$ there exists
a normal state $\omega$, over $\mathfrak{A}$, such that $\omega(A)\neq0.$
\end{enumerate}
\end{thm}

\begin{defn}
Let $(\cI,\le)$ be a directed set. (Sometimes we assume that each
pair $\alpha,\beta$ in \textit{$\cI$} has a least upper bound $\SYM{\alpha\vee\beta}{\alpha\vee\beta}$.)
A symmetric binary relation $\cauindep$ on $\cI$ is said to be an
\termi{orthogonality relation} if
\begin{enumerate}
\item if $\alpha\in\cI$ then there is a $\beta\in\cI$ with $\alpha\cauindep\beta,$
\item if $\alpha\leq\beta$ and $\beta\cauindep\gamma$ then $\alpha\cauindep\gamma,$
\item if $\alpha\cauindep\beta$ and $\alpha\cauindep\gamma$ then there
exists a $\delta\in\cI$ such that $\alpha\cauindep\delta$ and $\delta\geq\beta,\gamma.$
\end{enumerate}
\end{defn}

The orthogonality relation $\cauindep$ is intended to mean the causal
independence. %
For example, let \textit{$\cI$} be the set of bounded open subsets
of the Minkowski spacetime $\R^{1,3}$, together with
\begin{itemize}
\item $\alpha\le\beta$ iff $\alpha\subset\beta$,
\item $\alpha\cauindep\beta$ iff $\alpha$ and $\beta$ are spacelike separated.
\end{itemize}
Then we see that $\cauindep$ is an orthogonality relation of $(\cI,\le)$,
and $\alpha\vee\beta=\alpha\cup\beta$.

However, I would like to allow also some non-spacetime interpretations
of $\cI$. For example, consider the non-relativistic QM of an electron
on the space $\R^{3}$, formulated on the Hilbert space $L^{2}(\R^{3})\otimes\C^{2}$,
where $\C^{2}$ describes the spin states of it. Let $\alpha$ denote
this physical system. If we focus solely on the spin of the electron,
and ignore the position (and the momentum) of it, the Hilbert space
becomes $\C^{2}$. Let $\beta$ denote the latter system. %
If we ignore the spin, the suitable Hilbert space will be $L^{2}(\R^{3})$.
Let $\gamma$ denote the last system. Then we may write $\beta\le\alpha$
and $\gamma\le\alpha$. Perhaps, in this case, it would be more appropriate
to call $\alpha$, $\beta$ and $\gamma$ as \termi{scopes} (of information,
description, or observation, etc), rather than physical systems. Let%
\[
\cH_{\alpha}:=L^{2}(\R^{3})\otimes\C^{2},\qquad\cH_{\beta}:=\C^{2},\qquad\cH_{\gamma}:=L^{2}(\R^{3}),\qquad\cK_{\alpha\beta}:=L^{2}(\R^{3}),\qquad\cK_{\alpha\gamma}:=\C^{2},
\]
and $\fA_{\iota}:=\Bdd(\cH_{\iota}),$ $\iota=\alpha,\beta,\gamma$.
We see $\cH_{\alpha}=\cH_{\beta}\otimes\cK_{\alpha\beta}=\cH_{\gamma}\otimes\cK_{\alpha\gamma}$.
Thus $\fA_{\beta}$ is canonically embedded into $\fA_{\alpha}$ by
$x\mapsto x\otimes\bbOne$, $x\in\fA_{\beta}$, and hence $\fA_{\beta}$
can be viewed as a subalgebra of $\fA_{\alpha}$. Similarly, $\fA_{\gamma}$
can be viewed as a subalgebra of $\fA_{\alpha}$. This suggests the
concept of scope-dependent algebra of observables.
\begin{defn}[{See \cite[Sec.2.6]{BR87}}]
\label{def:causalNet} Let $(\cI,\le)$ be a directed set with the
orthogonality relation $\cauindep$. A net of {*}-algebras $\{\mathfrak{A}_{\alpha}\}_{\alpha\in\cI}$
is a \termi{causal net} (of algebras of local observables) if
\begin{enumerate}
\item if $\alpha\geq\beta$ then $\mathfrak{A}_{\alpha}\supset\mathfrak{A}_{\beta}$;
\item the algebras $\mathfrak{A}_{\alpha}$ have a common identity $\bbOne$;
\item $[\fA_{\alpha},\fA_{\beta}]=\{0\}$ if $\alpha\cauindep\beta$.
\end{enumerate}
\end{defn}

In the following, assume that each local algebra $\fA_{\alpha}$ is
a $C^{*}$-algebra. 

Each $\fA_{\alpha}$ is called a \termi{local algebra} or an \termi{algebra of local observables},
and $\mathfrak{A}:=\bigcup_{\alpha\in\cI}\mathfrak{A}_{\alpha}$ is
called the \termi{quasi-local algebra}. If each $\fA_{\alpha}$ is
a $C^{*}$-algebra, define the quasi-local algebra by $\mathfrak{A}=\ol{\bigcup_{\alpha\in\cI}\mathfrak{A}_{\alpha}}$,
where the bar denotes the uniform closure. In this case, $\fA$ is
also called the \termi{algebra of quasi-local observables}. (Note
that every self-adjoint element of $\bigcup_{\alpha\in\cI}\mathfrak{A}_{\alpha}$
is called a local observable, but $\bigcup_{\alpha\in\cI}\mathfrak{A}_{\alpha}$
itself is not called ``the algebra of local observables''.)

In non-relativistic QM, usually we can set $\fA_{\alpha}:=\Bdd(\cH_{\alpha})$,
where $\{\cH_{\alpha}|\alpha\in\cI\}$ is a net of Hilbert spaces.
However, in this case the inclusion relation $\alpha\le\beta\then\fA_{\alpha}\subset\fA_{\beta}$
is not evident, until the relation between $\cH_{\alpha}$ and $\cH_{\beta}$
is specified. Instead it is more convenient to consider a single Hilbert
space $\cH$ and the subalgebras $\fA_{\alpha}$ of $\Bdd(\cH)$ satisfying
$\alpha\le\beta\then\fA_{\alpha}\subset\fA_{\beta}$, and suppose
that each $\fA_{\alpha}$ is a von Neumann algebra, and further a
factor of type I.

In AQFTM \cite{Haa96,Araki99}, where each element of $\cI$ is a
bounded spacetime region, there seems to be a consensus that we can
assume that every $\fA_{\alpha}$ is a $W^{*}$-algebra, and furthermore
it is an approximately finite-dimensional (AFD, also called hyperfinite)
${\rm III}_{1}$ factor, which is unique up to equivalence. A good
review on this observation is found in \cite{Hal06}. See also \cite{Hor1990,Sum1990,Yng2005}.
We can define the canonical representation $(\pi_{\vac},\cH_{\vac})$
of the quasi-local algebra $\mathfrak{A}=\ol{\bigcup_{\alpha\in\cI}\mathfrak{A}_{\alpha}}$,
that is, the GNS representation with respect to the vacuum state $\omega_{\vac}$
on $\fA$. Therefore, each $\fA_{\alpha}$ can be viewed as a subalgebra
of $\Bdd(\cH_{0})$.

Yngvason \cite{Yng2005} emphasizes the conceptual importance of the
fact that the local algebra $\fA_{\alpha}$ is of type III. On the
other hand, the fact that $\fA_{\alpha}$ is AFD implies that it can
be well approximated by its subalgebras of type I. That is, within
some ``scopes of information'', approximately $\fA_{\alpha}$ can
be viewed as of type I. An actual/practical situation is as follows:
Sometimes we focus solely on a finite number of quantities in the
local system $\fA_{\alpha}$ to describe a physical law. In this case,
we consider a finite number of projection-valued measures (PVM's)
$E_{k}:\Borel(\R)\to\Proj(\fA_{\alpha})$, $k=1,...,n$ (or more generally
positive-operator-valued measures), and the $W^{*}$-subalgebra of
$\fA_{\alpha}$ generated by $\{E_{k}(X)|X\in\Borel(\R),k=1,...,n\}$,
which often is of type I.

For example, consider a neutral free scalar field (Klein\textendash Gordon
(KG) field) $\phi(f)$ ($f\in C_{\cpt}^{\infty}(\R^{1,3})$, the space
of compactly supported smooth functions on the Minkowski space $\R^{1,3}$),
as symmetric operators defined on a dense subspace of $\cH$. Let
$\cI$ be the set of all finite subsets of $C_{\cpt}^{\infty}(\R^{1,3})\setminus\{0\}$.
For any $\alpha,\beta\in\cI$, require that (1) $\alpha\le\beta\Leftrightarrow\alpha\subset\beta$,
and (2) $\alpha\cauindep\beta\ \Leftrightarrow$ $\supp(\alpha)$
and $\supp(\beta)$ are spacelike separated, where $\supp(\alpha):=\bigcup_{f\in\alpha}\supp(f)$.
For $f\in C_{\cpt}^{\infty}(\R^{1,3})$, let $E_{f}:\Borel(\R)\to\Proj(\cH)$
denote the spectral PVM of (the self-adjoint extension of) $\phi(f)$.
For $\alpha\in\cI$, let $\fA_{\alpha}$ be the von Neumann algebra
generated by $\{E_{f}(X)|X\in\Borel(\R),\,f\in\alpha\}$. Then $\fA_{\alpha}$
is always of type I, and often a factor.

Let $\SYM{\Proj(\fA)}{Proj}$ denote the set of projections in $\fA$.
$\Proj(\fA)$ is ordered by $e\le f$ iff $ef=e=fe$, $e,f\in\Proj(\fA)$.
\begin{defn}
Let $\fA$ be a $W^{*}$-algebra. A projection $e\in\fA$ is called
\termi{atomic} if it is minimal in $\Proj(\fA)\setminus\{0\}$. Let
$\SYM{\AProj(\fA)}{AProj}$ denote the set of atomic projections in
$\fA$. The $\fA$ is called \termi{atomic} if every nonzero projection
in $\fA$ has an atomic subprojection, that is, $\forall e\in\Proj(\fA)\setminus\{0\},$
$\exists f\in\AProj(\fA),$ $ef=f$. The $\fA$ is said to be \termi{$\sigma$-finite}
(or \termi{countably decomposable}) if it admits at most countably
many orthogonal projections.
\end{defn}

Every atomic $W^{*}$-algebra is of type I \cite[p.298]{Tak2002}.
Although the following Proposition \ref{prop:sigma-finite-atomic}
is a standard fact, its proof seems to be found only rarely in the
literature. Hence I will give a proof in Appendix.
\begin{prop}
\label{prop:sigma-finite-atomic}Let $\fR$ be a $W^{*}$-algebra.
Then, the following conditions are equivalent.
\begin{enumerate}
\item $\fR$ is $\sigma$-finite and atomic.
\item $\fR\cong\bigoplus_{i=1}^{n}\Bdd(\cH_{i}),$ where $n\in\N\cup\{\infty\}$,
and each $\cH_{i}$ is a separable Hilbert space.
\end{enumerate}
\end{prop}

The fact that every local algebra in the sense of Haag\textendash Araki\textendash Kastler
is AFD%
{} suggests the following

\def\algorithmname{Working Hypothesis}
\begin{lyxalgorithm}
[finitist version]\label{hyp:causal-finitist}In each of three quantum
theories (non-relativistic QM, QFTM and QFTCS), there exist a directed
set of suitable ``scopes'' $(\cI,\le)$ with an orthogonality relation
$\cauindep$ on $\cI$, and a causal net $\{\fA_{\alpha}|\alpha\in\cI\}$
which can describe the quantum physics, such that each local algebra
$\fA_{\alpha}$ is a finite-dimensional factor $W^{*}$-algebra, that
is, $\fA_{\alpha}\cong\Mat(n_{\alpha},\C)$ ($n_{\alpha}\times n_{\alpha}$
matrix algebra) for some $n_{\alpha}\in\N$.
\end{lyxalgorithm}

The above hypothesis \ref{hyp:causal-finitist} is conceptually intriguing,
and also more practically it might relate to computational physics
and lattice field theory. However, even the canonical commutation
relation (CCR) $[Q,P]=\im\hbar$ (precisely, the Weyl form of the
CCR) cannot be represented in a finite-dimensional $W^{*}$-algebra,
and so one is forced to consider some kind of finite-dimensional approximation
of the CCR. Thus instead I shall consider the following weaker hypotheses:
\begin{lyxalgorithm}
[type I factor version]\label{hyp:causal-typeIfactor}Similar to
Working Hypothesis \ref{hyp:causal-finitist}, but each $\fA_{\alpha}$
is a $\sigma$-finite type I factor $W^{*}$-algebra, that is, $\fA_{\alpha}\cong\Bdd(\cH_{\alpha})$
for some separable Hilbert space $\cH_{\alpha}$.
\end{lyxalgorithm}

\begin{lyxalgorithm}
[atomic version]\label{hyp:causal-atomic}Similar to Working Hypothesis
\ref{hyp:causal-finitist}, but each $\fA_{\alpha}$ is a $\sigma$-finite
atomic $W^{*}$-algebra. 
\end{lyxalgorithm}

Clearly, these Working Hypotheses satisfy \ref{hyp:causal-finitist}
$\then$ \ref{hyp:causal-typeIfactor} $\then$ \ref{hyp:causal-atomic}.
In all three Working Hypotheses \ref{hyp:causal-finitist}-\ref{hyp:causal-atomic},
each local algebra $\fA_{\alpha}$ has the canonical trace $\Tr_{\alpha}$,
which is faithful, semifinite and normal, satisfying $\Tr E=1$ for
all $E\in\AProj(\fA_{\alpha})$. Note that a $W^{*}$-algebra of type
III has no faithful semifinite normal trace \cite[Chapter V, Theorem 2.15]{Tak2002}.

In all these three hypotheses, every local algebra is assumed to be
$\sigma$-finite, and so I use ``a $W^{*}$-algebra'' to refer to
a $\sigma$-finite $W^{*}$ algebra in the following, unless otherwise
stated.

\section{Global state}

\label{sec:Global-state}

\subsection{QFTCS and the vacuum state}

A state on the quasi-local algebra $\fA$ is called a \termi{global state}.
Typically, a global state is understood as a state of the whole Universe,
and so practically we cannot prepare any global state. Thus the concept
of global state has little empirical meaning; Probably, a ``physical
law'' of the form
\begin{equation}
\text{If the Universe is prepared in the global state }\omega\text{, then ...}\label{eq:ifuniv}
\end{equation}
itself cannot be verified/falsified experimentally. Nevertheless,
the concept of global state plays a fundamental role in QFTM, because
any QFTM has a distinguished or ``privileged'' global state, namely,
the vacuum state. (Here I exclude thermal QFTM (at nonzero temperature),
which has no vacuum. I abbreviate ``thermal QFTM at temperature $T$''
to \termi{QFTM$_{T}$}, and ``thermal QFTM at a positive temperature''
to \termi{QFTM$_{>0}$}. For the rigorous nontrivial models of $(1+1)$-dimensional
QFTM$_{>0}$, see \cite{GJ2005,GJ2005b}. The usual QFTM with a (unique)
vacuum is sometimes called \emph{QFTM at zero temperature}. I abbreviate
it to QFTM$_{0}$, but ``QFTM'' always refers to QFTM$_{0}$ unless
otherwise stated, in this paper.)%
{} Thus the concept of %
global state appears to be almost indispensable in QFTM. Notice that
even if a statement of the type (\ref{eq:ifuniv}) itself might be
empirically meaningless, the statement can have some empirical consequences.
This situation is somewhat analogous to that of the Maxwell equation
in terms of the 4-vector potential $A_{\mu}$:
\begin{equation}
\di^{\nu}\di_{\nu}A_{\mu}-\di_{\mu}\di^{\nu}A_{\nu}=0.\label{eq:Maxwell}
\end{equation}
Since $A_{\mu}$ is unobservable, one has some reason to think that
(\ref{eq:Maxwell}) itself is empirically meaningless. Nevertheless,
extremely rich empirical consequences are deduced from (\ref{eq:Maxwell}).
In QFTM, the \emph{Wightman functions} (i.e., the vacuum expectation
values) play an analogous role to (\ref{eq:Maxwell}).

From a point of view of ``pragmatic realism'', one might argue as
follows:
\begin{claim}
\textbf{Pragmatic realism for the (global) vacuum state}: We are permitted
to regard the (global) vacuum state as a physical \emph{reality} (even
though the actual universe in which we live is not in a vacuum state),
because the concept of vacuum is \emph{indispensable} in QFTM.
\end{claim}

However, this argument based on the indispensability of the concept
of vacuum loses all ground in QFTCS, where the universal concept of
vacuum does not exist. More precisely, in QFTCS, we cannot even define
the concept of vacuum%
{} in a generic spacetime, in a covariant (or observer-independent)
manner. It is not that the vacuum is definable but does not exist.
In non-relativistic QM and QFTM, a ground state (vector) can be defined
as a unit vector $\psi$ in $\cH$ such that $H\psi=\lambda\psi$,
where $H$ is the Hamiltonian and $\lambda\in\R$ is the infimum of
the spectrum of $H$, whereas such vector $\psi$ can be nonexistent.
In QFTM, the ground state is called the vacuum; Here the notion of
Hamiltonian is not Lorentz invariant (i.e., it is inertial-observer-dependent),
but the vacuum becomes Lorentz invariant. In QFTCS, the notion of
inertial observer does not make sense in general, and so one might
wish to search for an alternative definition of vacuum which is observer-independent
in some other sense. However, such a definition seems to be impossible
(see Fredenhagen \& Rejzner \cite{FR2016}).

The negative implication of QFTCS on the concept of global state is
even stronger. As mentioned above, in QFTM, there exists a canonical
(``privileged'') representation of the quasi-local algebra $\fA$,
i.e., the GNS representation $(\pi_{\vac},\cH_{\vac})$ with respect
to the vacuum state $\omega_{\vac}$ on $\fA$. It is called the \termi{vacuum representation}.
A possible interpretation of this fact is as follows: A state $\omega$
on $\fA$ is physically realizable if and only if there exists a density
matrix $\rho$ on $\cH_{\vac}$ such that $\omega(A)=\Tr\rho\pi_{\vac}(A)$
for all $A\in\fA$. %
{} Thus in QFTM, there is a clear distinction between the physical states
(i.e., physically realizable states) and the non-physical (or ``fictional'')
states. (Exactly speaking, rather we should say that usual QFTM (i.e.,
QFTM$_{0}$) has the vacuum \emph{as a physical state}. Similarly,
we should say that thermal QFTM (i.e., QFTM$_{>0}$) does not have
the vacuum \emph{as a physical state}; Some kind of QFTM$_{>0}$ might
have the vacuum \emph{as a non-physical state}.)

On the other hand, in QFTCS, there is no such clear distinction; There
is no privileged representation of $\fA$. Hence the concept of global
state becomes more ``fictional'', because we cannot distinguish
the physical states from the fictional states. In fact, already in
Minkowski spacetime, the distinction between QFTM$_{0}$ and QFTM$_{>0}$
seems to be observer-dependent. See the considerations on the conceptual
problems on the Unruh effect \cite{Ear2011,Yam2026a}.

\subsection{Heisenberg picture in QFTM}

Even if we confine our discussion to QFTM, the empirical significance
of global states remains to be questionable. Consider the ``physical
law'' in QFTM of the following form:
\begin{equation}
\text{If the system }\text{is prepared in the state }\omega\text{ at time }t_{0},\text{ then ...}\label{eq:prepare-t0}
\end{equation}
I will argue that the meaning of (\ref{eq:prepare-t0}) is far from
clear, both empirically and mathematically.

Recall that a QFTM is usually formulated in the Heisenberg picture,
where the notion of quantum state is time-independent. However, in
non-relativistic QM, it might be possible to understand the Heisenberg
picture to be a mere rephrasing of the Schr\"{o}dinger picture, and
so any notion defined in the latter can be used also in the former
with a slight modification. On the other hand, the Heisenberg picture
appears to be more essential in QFTM. If one wish to consider some
kind of ``Schr{\"o}dinger picture version'' of QFTM, one might
be inclined to think that ``the state at time $t$'' refers to the
state of ``the field at time $t$''. However, this will not be the
correct understanding, because it is conjectured that the \emph{sharp-time
field} (also called the \emph{time-sliced field}) at each time $t$
cannot be defined in $(3+1)$-dimensional spacetime when the field
is interacting. Thus ``the state at time $t$'' in a generic QFTM
cannot be the state of ``the field at time $t$''. But in this case,
what object's state is it?

However, I would not go so far as to say that the statement (\ref{eq:prepare-t0})
is entirely meaningless. For example, some fragments of QBism might
be employed to salvage the concept of state at a specific time $t$.
(This is ironic, since basically QBism is skeptical of the concept
of state.) That is, we can understand ``the state at time $t$''
as the ``subjective state at $t$'' of an ``agent'', i.e., the
state of an agent's belief or knowledge at time $t$. If one dislikes
the subjectivism, instead one could regard the agent as a computer
(with an enough memory), and consider the ``informational state''
of the agent at time $t$.

In any way, I do not believe that it is necessary to salvage the concept
of state at $t$, since I consider the concept to be dispensable.

\section{Local state}

\label{sec:Local-state}

Assume Hypothesis \ref{hyp:causal-typeIfactor} (type I factor version)
rather than \ref{hyp:causal-atomic} (atomic version) for simplicity,%
{} so that $\fA_{\alpha}\cong\Bdd(\cH_{\alpha})$ for some separable
Hilbert space $\cH_{\alpha}$.

A normal state on a local algebra $\fA_{\alpha}$ ($\alpha\in\cI$)
is called a \termi{local state} on $\fA_{\alpha}$, or a \termi{$\alpha$-local state}.
If $\omega:\fA\to\C$ is a global state on $\fA$, and if the restriction
$\omega_{\alpha}:=\omega|_{\fA_{\alpha}}$ is normal, then $\omega_{\alpha}$
is a $\alpha$-local state. However, there are doubts as to how much
empirical significance the concept of global state holds, as mentioned
above. Hence we examine the empirical significance of the concept
of local state \emph{without relying on that of global state}.

Let $\Tr_{\alpha}$ denote the canonical trace on $\fA_{\alpha}$,
that is, if $\fA_{\alpha}$ is identified with $\Bdd(\cH_{\alpha})$,
$\Tr_{\alpha}$ is the usual trace on $\Bdd(\cH_{\alpha})$. %
In this case, $\omega_{P}(A):=\Tr_{\alpha}PA$, $A\in\fA_{\alpha}$
defines a local pure state $\omega_{P}$ on $\fA_{\alpha}$. Note
that the concept of local state is ``scope-dependent'', since $\fA_{\alpha}$
depends on $\alpha\in\cI$.\footnote{Although Rovelli's relational quantum mechanics \cite{Rov1996} stress
the \emph{observer}-dependence of the concept of state, rather I would
like to understand his arguments as a manifestation of the \emph{scope}-dependence.} Especially, if $\beta\le\alpha$, then the restricted local state
$\omega_{P}|_{\fA_{\beta}}$ on $\fA_{\beta}$ is not a pure local
state, in general. Thus the notion of ``pureness'' of local states
is highly scope-dependent. 

Many authors have been arguing that although a mixed state might be
merely a measure of one's ignorance, a pure state is a physical reality
(especially when it is expressed as a wave function); roughly speaking,
that pure states are physically real, but non-pure states are not.
I call such an argument a ``\termi{pure-state-specific realism}''.
However, in the framework of causal net, there is no absolute distinction
between pure local states and non-pure local states. Thus any pure-state-specific
realism seems to be logically weak and fragile, because it presupposes
the existence of such an absolute distinction. Moreover, although
some kind of absolute notion of pure \emph{global} state might be
defined mathematically, we cannot distinguish physically realizable
(pure or mixed) global states from non-physical ``fictional'' global
states in QFTCS, as I argued in Section \ref{sec:Global-state}. No
pure-state-specific realist would argue that such a possibly fictional
state should be a physically reality; it is almost a self-contradiction.

By the assumption that $\fA_{\alpha}$ is an algebra of observables,
any projection in $\fA_{\alpha}$ can be realized as an experimental
procedure (yes-no type measurement), in principle. Thus we can prepare
the pure $\alpha$-local state $\omega_{P}$ for any $P\in\AProj(\fA_{\alpha})$.
Furthermore, any mixed local state on $\fA_{\alpha}$ can be prepared
by the mixing operation of pure $\alpha$-local states. Thus the concept
of local state is much less problematic than that of global state,
from the empirical point of view.

However the key points of my argument can be summarised as follows.

(1) If, contrary to Hypothesis \ref{hyp:causal-atomic}, $\fA_{\alpha}$
was a factor of type III, the empirical meaning of a normal state
$\omega$ on $\fA_{\alpha}$ becomes much more unclear. Namely, contrary
to the type I cases, it is not immediately clear how to describe/construct
the preparing procedure of a given state $\omega$, and also it is
not clear whether we can experimentally verify/falsify a ``physical
law'' of the form:
\begin{equation}
\text{If the system is prepared in the state }\omega,\text{ then ...}\label{eq:law-prepare-local}
\end{equation}

(2) When $\fA_{\alpha}$ is a type I factor by Hypothesis \ref{hyp:causal-typeIfactor}
(or an atomic $W^{*}$-algebra by Hypothesis \ref{hyp:causal-atomic}),
the concept of $\alpha$-local state can be defined with the concept
of (experimental) operation in $\fA_{\alpha}$, as mentioned above.
Hence, for example, (\ref{eq:law-prepare-local}) can be rephrased
as
\begin{equation}
\text{If the system is prepared by the operation }X,\text{ then ...}\label{eq:law-prepare-local-op}
\end{equation}
Thus, any argument that the concept of local state is indispensable
does not seem particularly convincing. Rather, it will be the concept
of \emph{operation} that is indispensable.

From (1) and (2) above, whether each $\fA_{\alpha}$ is a type I factor
or a type III factor, I doubt the ``empirical indispensability''
of the concept of local state; If $\fA_{\alpha}$ is of type III,
the empirical significance of the concept of local state is rather
doubtful, even if it is indispensable in the purely mathematical theory
of von Neumann algebras. If $\fA_{\alpha}$ is of type I, the concept
is much less problematic empirically, but can be replaced with the
more general concept of (local) operation.

\section{Non-relativistic QM without states}

\label{sec:NRQM-withoutStates}

If one wish to demonstrate that the concept of state is dispensable
in quantum physics, the most persuasive way will be to actually develop
quantum physics without employing the concept of state at all. If
the local algebras are assumed to be $\sigma$-finite atomic $W^{*}$-algebras
by Hypothesis \ref{hyp:causal-atomic}, the concept of local state
can be replaced with that of local operation in principle. Hence,
at first glance, this task does not seem too difficult.

However the conventional formulation of QFTM largely depends on a
privileged global state, i.e., the vacuum. Especially in the Wightman
formalism, all the physical laws are derived from the vacuum expectation
values. In QFTCS, we will need to fundamentally rethink this conventional
Wightman approach, and to construct an alternative language to describe
the physical laws. This might turn out to be an extremely hard task.
Thus the following considerations on QFTM and QFTCS will remain at
a preliminary level.

Although I shall not go into quantum statistical mechanics (QSM) in
this paper, here I will give a brief comment on it. It should be noted
that the concept of thermal equilibrium state (Gibbs state or more
general KMS state) plays a central role in QSM, similarly to the concept
of vacuum state in QFTM. Hence also it might be a very hard task to
develop a ``stateless formalism'' of QSM. However, recall that there
are some conceptual problems already in classical statistical mechanics,
stated as follows. Thermodynamics describes the macroscopic thermal
phenomena with which we are all familiar. In this sense, the concept
of ``thermodynamic state'' seems to be physically real. On the other
hand, (classical or quantum) statistical mechanics aims to give the
microscopic explanations of thermodynamical phenomena, including phase
transitions. However, the concept of thermal equilibrium state in
statistical mechanics, which is employed in such explanations, is
considerably idealized and unrealistic, in that it is justified only
through a procedure called a \emph{thermodynamic limit} (e.g., infinite
volume limit and/or infinite particle number limit). Thus, it is doubted
that macroscopic thermodynamical states can be identified with microscopic
thermal equilibrium states.%
{} Also note that thermal equilibrium states are rigorously defined
even in curved spacetime (\cite{HW2015}, see also \cite[Definition 9.11]{Ger2019}),
but such states seem to become conceptually far more problematic than
in non-relativistic QSM. See Yamashita \cite{Yam2026a}.%

Temporarily assume Hypothesis \ref{hyp:causal-typeIfactor} (type
I factor version) rather than \ref{hyp:causal-atomic} (atomic version)
for simplicity. Fix a ``scope'' $\alpha\in\cI$, and let $\fA_{\alpha}$
be a local $W^{*}$-algebra. Hence $\fA_{\alpha}$ can be identified
with $\Bdd(\cH_{\alpha})$ for some separable Hilbert space $\cH_{\alpha}$.
However note that if $\alpha\le\beta$, we have the relation $\fA_{\alpha}\subset\fA_{\beta}$,
but the relation between $\cH_{\alpha}$ and $\cH_{\beta}$ is not
specified \emph{a priori}. Hence the identification $\fA_{\alpha}\cong\Bdd(\cH_{\alpha})$
will become less convenient when $\alpha\in\cI$ is not fixed.

First let us consider non-relativistic QM. Causal nets of local algebras
are often used in so-called quantum physics with infinite degrees
of freedom, that is, (relativistic or non-relativistic) QFT and quantum
statistical mechanics \cite{Haa96,Araki99,BR87,BR97}. Furthermore,
if one wishes to understand the philosophy behind Rovelli's RQM through
the notion of causal net, ultimately all quantum physics should be
formulated with that notion. However, for most of practical purposes,
it is true that we can fix a ``scope'' $\alpha\in\cI$ and a local
algebra $\fA_{\alpha}$ to describe a given non-relativistic quantum
system of finite degrees of freedom. Thus fix $\alpha\in\cI$, and
for simplicity, assume $\fA_{\alpha}\cong\Bdd(\cH_{\alpha})$ for
some separable Hilbert space $\cH_{\alpha}$. Let $\Tr=\Tr_{\alpha}$
denote the canonical trace on $\fA_{\alpha}$. 

\subsection{\protect\label{subsec:A-spinless-particle}A spinless particle on
a line}

For each $t\in\R$, let $\quantity_{t}$ be a real-valued physical
quantity.%
{} We can see $\quantity_{t}$ as referring to the physically meaningful
noun phrase such as ``the position (of a particle) at time $t$'',
rather than a kind of mathematically well-defined object (such as
a real number or a set of real numbers).

Let $\Quantities=\{\quantity_{t}|t\in\R\}$. It is reasonable to assume
that the map $\R\to\Quantities$, $t\mapsto\quantity_{t}$ is injective
(and so bijective). For example, as to a harmonic oscillator with
the period $T>0$, the following will be a nontrivial physical law:
``the position at $t$ is equal to the position at $t+T$'', in
other words, ``the value of $\quantity_{t}$ is equal to the value
of $\quantity_{t+T}$''. On the other hand, if we assume $\quantity_{t}=\quantity_{t+T}$,
the above physical law reduces to the tautology ``the value of $\quantity_{t}$
is equal to the value of $\quantity_{t}$'', which is physically
vacuous. This is absurd, and hence we should assume $\quantity_{t}\neq\quantity_{t+T}$
to describe the nontrivial physical/empirical laws.%

For each $\R$-valued quantity $\quantity_{t}$, we assign a corresponding
PVM $E_{t}:\Borel(\R)\to\Proj(\fA_{\alpha})$, and moreover assume
that the time evolution $E_{0}\mapsto E_{t}$ ($t\in\R$) is unitary,
that is, there exists a one-parameter unitary group $\{U(t)|t\in\R\}\subset\fA_{\alpha}$
such that $E_{t}(X)=U(-t)E_{0}(X)U(t)$ for all $t\in\R$ and $X\in\Borel(\R)$.

In most of practical examples, the following holds:
\begin{assumption}
\label{assu:trace-class2}If both $X$ and $X'$ are bounded Borel
subsets of $\R$, then $E_{t}(X)E_{t'}(X')$ is trace-class for almost
all $(t,t')\in\R^{2}$.
\end{assumption}

Let $X,X'\in\Borel(\R)$ be bounded and non-null. For example, in
the case of a free particle, where $\quantity_{t}$ denotes the position
of the particle at $t$, we see $E_{t}(X)E_{t'}(X')$ is trace-class
if and only if $t\neq t'$. In the case of a harmonic oscillator with
period $T$, we see $E_{t}(X)E_{t'}(X')$ is trace-class if and only
if $t-t'\notin T\Z$.

Let $n\ge2$, $1\le k<n$, and $X_{1},...,X_{n}\in\Borel(\R)$ be
bounded. If $E_{t_{1}}(X_{1})E_{t_{2}}(X_{2})$ is trace-class by
Assumption \ref{assu:trace-class2}, then $E_{t_{1}}(X_{1})\cdots E_{t_{n}}(X_{n})$
is also trace-class. Define the \termi{prior conditional probability}
by
\begin{align}
 & \Prob(X_{k+1},...,X_{n};t_{k+1},...,t_{n}|X_{1},...,X_{k};t_{1},...,t_{k}):=\frac{\Tr(AB)^{*}AB}{\Tr A^{*}A},\label{eq:priorCondProb}\\
 & \qquad\qquad A:=E_{t_{1}}(X_{1})\cdots E_{t_{k}}(X_{k}),\quad B:=E_{t_{k+1}}(X_{k+1})\cdots E_{t_{n}}(X_{n}),\nonumber 
\end{align}
if the denominator is nonzero, i.e. $A\neq0$. Notice that this probability
is defined without the notion of quantum state. The empirical meaning
of (\ref{eq:priorCondProb}) appears to be clear only when $t_{i}$'s
are (anti-)time-ordered, i.e., $t_{1}<\cdots<t_{n}$ or $t_{1}>\cdots>t_{n}$.
If $t_{1}<\cdots<t_{n}$, the prior conditional probability (\ref{eq:priorCondProb})
is a \emph{predictive} probability; That is, in the sequential (composite)
measurement of the yes-no type measurements $E_{t_{1}}(X_{1}),...,E_{t_{n}}(X_{n})$
in this order, (\ref{eq:priorCondProb}) means the probability that
all $E_{t_{k+1}}(X_{k+1}),...,E_{t_{n}}(X_{n})$ give ``yes'', conditioned
that all $E_{t_{1}}(X_{1}),...,E_{t_{k}}(X_{k})$ gave ``yes''.
In other words, this is the probability that the value of $\quantity_{t_{i}}$
is measured to be in $X_{i}$ for all $i=k+1,...,n$, conditioned
that the value of $\quantity_{t_{j}}$ was measured to be in $X_{j}$
for all $j=1,...,k$.

Contrarily, if $t_{1}>\cdots>t_{n}$, the prior conditional probability
(\ref{eq:priorCondProb}) is a \emph{retrodictive} probability. Note
that although we do not frequently employ the retrodictive reasoning
in physics, it is indispensable in some situations, e.g., in cosmology,
especially the Big Bang theory.

Actually, the following much stronger property is (usually implicitly)
assumed for many cases in the physics literature, instead of Assumption
\ref{assu:trace-class2}:
\begin{assumption}
\label{assu:propagator}There exists a $\C$-valued function $K_{t,t'}(x,x')\equiv K(x,t;x',t')$,
called a \termi{propagator}, defined for all $x,x'\in\R$ and almost
all $(t,t')\in\R^{2}$,%
{} such that
\begin{enumerate}
\item $K_{t,t'}(x,x')$ is continuous in $x,x'$, if $K_{t,t'}$ is defined.
\item for any $n\ge2$, for almost all $(t_{1},...,t_{n})\in\R^{n}$, and
for all bounded Borel subsets $X_{1},...,X_{n}$ of $\R$, the following
holds:
\begin{equation}
\Tr E_{t_{1}}(X_{1})\cdots E_{t_{n}}(X_{n})=\int_{X_{1}}\d x_{1}\cdots\int_{X_{n}}\d x_{n}\ \prod_{i=1}^{n}K(x_{i},t_{i};x_{i+1},t_{i+1}),\qquad x_{n+1}:=x_{1},\ t_{n+1}:=t_{1}.\label{eq:TrEE=00003DintintK}
\end{equation}
\item $K(x,t;x',t')=\ol{K(x',t';x,t)},$
\item $\int_{\R}\d x'\ K_{t,t'}(x,x')K_{t',t''}(x',x'')=K_{t,t''}(x,x'').$ 
\end{enumerate}
\end{assumption}

Especially, we see
\[
\Tr E_{t_{1}}(X_{1})E_{t_{2}}(X_{2})=\int_{X_{1}}\d x_{1}\int_{X_{2}}\d x_{2}\ \left|K(x_{1},t_{1};x_{2},t_{2})\right|^{2}.
\]
Note that the conditions (1), (3) and (4) in Assumption \ref{assu:propagator}
are rather optional, since all the observable consequences are derived
from (2) only. If we set $\cH_{\alpha}=L^{2}(\R)$, and if the time
evolution is ``uniform'', that is, if the Hamiltonian is time-independent,
then $K_{t,t'}$ is nothing but the integral kernel of the time evolution
operator $U(t'-t)$.

As is well-known, a propagator for a free particle of mass $m$ is
given by
\begin{equation}
K_{t,t'}(x,x')=\left(\frac{m}{2\pi\im\hbar T}\right)^{1/2}\exp\left[\frac{\im m}{2\hbar T}(x'-x)^{2}\right],\qquad T:=t'-t,\label{eq:K-free}
\end{equation}
and a propagator for a harmonic oscillator is given by
\begin{equation}
K_{t,t'}(x,x')=\left(\frac{m\omega}{2\pi\im\hbar\sin\omega T}\right)^{1/2}\exp\left\{ \frac{\im m\omega}{2\hbar}\left[(x^{2}+x^{\prime2})\cot\omega T-2\frac{xx'}{\sin\omega T}\right]\right\} ,\qquad T:=t'-t.\label{eq:K-oscill}
\end{equation}
Assumption \ref{assu:propagator} suggests that the concept of quantum
state is dispensable to describe the empirical laws in non-relativistic
quantum mechanics; Instead the notion of propagator seems sufficient
to describe the empirical laws. However, it should be noted that the
notion of propagator itself, as well as the examples (\ref{eq:K-free})\textendash (\ref{eq:K-oscill}),
are not immediately given an empirical meaning. For example, we expect
that any empirical law for a free particle should be invariant under
the Galilei transformation $(x,t)\mapsto(x+vt,t)$. However, (\ref{eq:K-free})
is not Galilei invariant, and hence (\ref{eq:K-free}) itself cannot
be viewed as an empirical law, as I argued in Yamashita \cite{Yam2025}.
Eq.~(\ref{eq:K-free}) is understood to have some empirically redundant
information, something like a gauge fixing for the Galilei group.
This suggests that Assumption \ref{assu:propagator} is not satisfactory
as a fundamental axiom of quantum physics, though it has a considerable
practical/computational advantage.

On the other hand, (\ref{eq:TrEE=00003DintintK}) indicates that the
function
\begin{equation}
\SYM{\kappa(\vec{x},\vec{t})}{kappa}\equiv\kappa_{n}(\vec{x},\vec{t}):=\prod_{i=1}^{n}K(x_{i},t_{i};x_{i+1},t_{i+1}),\qquad\vec{x}:=(x_{1},...,x_{n}),\quad\vec{t}:=(t_{1},...,t_{n}),\label{eq:def:kappa}
\end{equation}
shares all the symmetries (such as the Galilei covariance for a free
particle) with the l.h.s.~$\Tr E_{t_{1}}(X_{1})\cdots E_{t_{n}}(X_{n})$,
which is directly related to the observable physical laws in terms
of prior conditional probability. Hence the function $\kappa(\vec{x},\vec{t})$
is expected to be ``more empirical'' than the propagator $K(x,t;x',t')$,
that is, $\kappa(\vec{x},\vec{t})$ has very little empirical redundancy.
(But note that $\kappa(\vec{x},\vec{t})$ and its complex conjugate
$\ol{\kappa(\vec{x},\vec{t})}$ seem to lead to the same empirical
consequences, and hence the imaginary part of $\kappa(\vec{x},\vec{t})$
will not be observable.) Precisely, we see from (\ref{eq:priorCondProb})
that it is sufficient to know $\kappa(\vec{x},\vec{t})$ only when
$\vec{t}$ is of the form
\begin{equation}
\vec{t}=(t_{1},...,t_{2k-2})=(t_{1},...,t_{k},t_{k-1},...,t_{2}),\quad\text{i.e., }(t_{k+1},...,t_{2k-2})=(t_{k-1},...,t_{2}),\qquad k\in\N,\label{eq:t(2k-2)}
\end{equation}
to calculate the prior conditional probabilities, because
\[
\Tr(F_{1}\cdots F_{k})^{*}(F_{1}\cdots F_{k})=\Tr F_{1}\cdots F_{k}F_{k-1}\cdots F_{2},\qquad F_{i}:=E_{t_{i}}(X_{i}).
\]
For example, when $k=2$, let $\vec{x}:=(x_{1},x_{2})$ and $\vec{t}:=(t_{1},t_{2})$.
Then we have
\[
\kappa_{2}(\vec{x},\vec{t})=\prod_{i=1}^{2}K(x_{i},t_{i};x_{i+1},t_{i+1})=K(x_{1},t_{1};x_{2},t_{2})K(x_{2},t_{2};x_{1},t_{1})=\left|K(x_{1},t_{1};x_{2},t_{2})\right|^{2}.
\]
When $k=3$, let
\[
\vec{t}=(t_{1},...,t_{4})=(t_{1},t_{2},t_{3},t_{2}),\ t_{5}:=t_{1},\qquad x_{1},x_{2},x_{3},x_{4}\in\R,\ x_{5}:=x_{1},\ x_{1}':=x_{1},\ x_{2}':=x_{4},\ x_{3}':=x_{3}.
\]
Then we have
\begin{align*}
\kappa_{4}(\vec{x},\vec{t})=\prod_{i=1}^{4}K(x_{i},t_{i};x_{i+1},t_{i+1}) & =K(x_{1},t_{1};x_{2},t_{2})K(x_{2},t_{2};x_{3},t_{3})K(x_{3},t_{3};x_{4},t_{2})K(x_{4},t_{2};x_{1},t_{1})\\
 & =K(x_{1},t_{1};x_{2},t_{2})K(x_{2},t_{2};x_{3},t_{3})K(x_{3}',t_{3};x_{2}',t_{2})K(x_{2}',t_{2};x_{1}',t_{1})\\
 & =\prod_{i=1}^{2}K(x_{i},t_{i};x_{i+1},t_{i+1})\ol{\prod_{i=1}^{2}K(x_{i}',t_{i};x_{i+1}',t_{i+1})}
\end{align*}
Generally, for $k\ge3$, assume (\ref{eq:t(2k-2)}) and let
\[
x_{1},...,x_{2k-2}\in\R,\ x_{2k-1}:=x_{1},\ x_{1}':=x_{1},\ x_{k}':=x_{k},\ x_{i}':=x_{2k-i}\text{ for }i=2,...,k-1.
\]
Then we have
\begin{equation}
\kappa_{2k-2}(\vec{x},\vec{t})=\prod_{i=1}^{2k-2}K(x_{i},t_{i};x_{i+1},t_{i+1})=\prod_{i=1}^{k-1}K_{i}\ol{K_{i}'},\label{eq:kappa(2k-2)KK}
\end{equation}
where
\[
K_{i}:=K(x_{i},t_{i};x_{i+1},t_{i+1}),\qquad K_{i}':=K(x_{i}',t_{i};x_{i+1}',t_{i+1}).
\]

\begin{prop}
\label{prop:kappa24}Let $\kappa$ be defined by (\ref{eq:def:kappa}),
and suppose (\ref{eq:t(2k-2)}). Let
\[
\kappa_{2}^{(i)}:=\kappa_{2}((x_{i},x_{i+1}'),(t_{i},t_{i+1})),\qquad\kappa_{4}^{(i)}:=\kappa_{4}((x_{i},x_{i+1},x_{i+2}',x_{i+1}'),(t_{i},t_{i+1},t_{i+2},t_{i+1})).
\]
Then we have
\begin{equation}
\kappa_{2k-2}(\vec{x},\vec{t})=\frac{\prod_{i=1}^{k-2}\kappa_{4}^{(i)}}{\prod_{j=2}^{k-2}\kappa_{2}^{(j)}},\label{eq:kappa(2k-2)=00003Dk/k}
\end{equation}
and hence the two functions $\kappa_{2}$ and $\kappa_{4}$ determine
$\kappa_{2k-2}$ for all $k\ge4$.%
\end{prop}

\begin{proof}
Let
\[
L_{i}:=K(x_{i},t_{i};x_{i+1}',t_{i+1}).
\]
Then we find
\[
\kappa_{4}^{(1)}=K_{1}\ol{K_{1}'K_{2}'}L_{2},\qquad\kappa_{4}^{(k-2)}=K_{k-2}K_{k-1}\ol{K_{k-1}'}\ol{L_{k-2}},
\]
and
\[
\kappa_{4}^{(i)}=K_{i}\ol{K_{i+1}'}L_{i+1}\ol{L_{i}},\qquad\text{for }i=2,...,k-3.
\]
In fact, we have
\begin{align*}
\kappa_{4}^{(1)} & =\kappa((x_{1},x_{2},x_{3}',x_{2}'),(t_{1},t_{2},t_{3},t_{2}))\\
 & =K(x_{1},t_{1};x_{2},t_{2})K(x_{2},t_{2};x_{3}',t_{3})K(x_{3}',t_{3};x_{2}',t_{2})K(x_{2}',t_{2};x_{1}',t_{1})\\
 & =K_{1}\ol{K_{1}'K_{2}'}L_{2},
\end{align*}
and similarly for $\kappa_{4}^{(k-2)}$ and $\kappa_{4}^{(i)}$ ($i=2,...,k-3$).%
{} Thus we find
\begin{align*}
\prod_{i=1}^{k-2}\kappa_{4}^{(i)} & =\kappa_{4}^{(1)}\left(\prod_{i=2}^{k-3}\kappa_{4}^{(i)}\right)\kappa_{4}^{(k-2)}\\
 & =K_{1}\ol{K_{1}'K_{2}'}L_{2}\left(\prod_{i=2}^{k-3}K_{i}\ol{K_{i+1}'}L_{i+1}\ol{L_{i}}\right)K_{k-2}K_{k-1}\ol{K_{k-1}'}\ol{L_{k-2}}\\
 & =\left(\prod_{i=1}^{k-1}K_{i}\ol{K_{i}'}\right)\left(\prod_{j=2}^{k-2}L_{j}\ol{L_{j}}\right)\\
 & =\kappa_{2k-2}(\vec{x},\vec{t})\left(\prod_{j=2}^{k-2}\kappa_{2}^{(j)}\right).\qquad\text{(by (\ref{eq:kappa(2k-2)KK})})
\end{align*}
Thus (\ref{eq:kappa(2k-2)=00003Dk/k}) follows.%
\end{proof}
Considering Proposition \ref{prop:kappa24}, let us temporarily forget
Assumption \ref{assu:propagator}, and instead examine the following
Assumption \ref{assu:kappa24}. This seems considerably more complicated
and less intuitive than Assumption \ref{assu:kappa24}. Of course,
it is desirable for a fundamental axiom of quantum mechanics to be
simple and intuitive, and so I do not intend to argue that Assumption
\ref{assu:kappa24} is a fundamental axiom. Our main purpose is to
reduce the ``empirical redundancy'' of the notion of propagator
in Assumption \ref{assu:propagator}, to search for a better means
of describing the empirical law.%
{} %

\begin{assumption}
\label{assu:kappa24}There exist two functions $\kappa_{2;t_{1},t_{2}}$
and $\kappa_{4;t_{1},t_{2},t_{3}}$ satisfying the following conditions:
\begin{enumerate}
\item $\kappa_{2;t_{1},t_{2}}(x_{1},x_{2})\equiv\kappa_{2}(x_{1},x_{2};t_{1},t_{2})$
is a $[0,\infty)$-valued function defined for almost all $(t_{1},t_{2})\in\R^{2}$
with $t_{1}<t_{2}$ and for all $(x_{1},x_{2})\in\R^{2}$ such that
\[
\Tr\left(F_{1}F_{2}\right)=\int_{X_{1}}\d x_{1}\int_{X_{2}}\d x_{2}\ \kappa_{2;t_{1},t_{2}}(x_{1},x_{2}),\qquad F_{i}:=E_{t_{i}}(X_{i}),\ X_{i}\in\Borel(\R):\text{bounded, }i=1,2.
\]
\item $\kappa_{4;t_{1},t_{2},t_{3}}(x_{1},x_{2},x_{3},x_{2}')\equiv\kappa_{4}(x_{1},x_{2},x_{3},x_{2}';t_{1},t_{2},t_{3},t_{2})$
is a $\C$-valued function defined for almost all $(t_{1},t_{2},t_{3})\in\R^{3}$
with $t_{1}<t_{2}<t_{3}$ and for all $(x_{1},x_{2},x_{3},x_{2}')\in\R^{4}$
such that
\begin{align}
\Tr\left(F_{1}F_{2}F_{3}\right)^{*}\left(F_{1}F_{2}F_{3}\right) & \equiv\Tr F_{1}F_{2}F_{3}F_{2}\nonumber \\
 & =\int_{X_{1}}\d x_{1}\int_{X_{2}}\d x_{2}\int_{X_{3}}\d x_{3}\int_{X_{2}}\d x_{2}'\ \kappa_{4;t_{1},t_{2},t_{3}}(x_{1},x_{2},x_{3},x_{2}'),\label{eq:F1F2F3F2=00003Dkappa4}
\end{align}
where $F_{i}:=E_{t_{i}}(X_{i})$, and each $X_{i}\in\Borel(\R)$ is
bounded.
\item The following holds for all $k\ge4$:
\begin{align}
 & \Tr\left(F_{1}\cdots F_{k}\right)^{*}\left(F_{1}\cdots F_{k}\right)=\Tr F_{1}\cdots F_{k}F_{k-1}\cdots F_{2}\nonumber \\
 & \qquad=\int_{X_{1}}\d x_{1}\cdots\int_{X_{k}}\d x_{k}\int_{X_{k-1}}\d x_{k-1}'\cdots\int_{X_{2}}\d x_{2}'\frac{\prod_{i=1}^{k-2}\kappa_{4;t_{i},t_{i+1},t_{i+2}}(x_{i},x_{i+1},x_{i+2}',x_{i+1}')}{\prod_{i=2}^{k-2}\kappa_{2;t_{1},t_{2}}(x_{i},x_{i+1}')}.\label{eq:TrF1Fk=00003Dkappa}
\end{align}
\end{enumerate}
\end{assumption}

Note that the imaginary part of $\kappa_{4;t_{1},t_{2},t_{3}}$ cancels
in the integral in (\ref{eq:F1F2F3F2=00003Dkappa4}), and hence only
the real part of $\kappa_{4;t_{1},t_{2},t_{3}}$ contributes to the
value of $\Tr F_{1}F_{2}F_{3}F_{2}$. This implies that the real part
$\Re\kappa_{4;t_{1},t_{2},t_{3}}$, together with $\kappa_{2;t_{1},t_{2}}$,
is essentially an observable quantity. However, $\kappa_{4;t_{1},t_{2},t_{3}}$
should be taken to be $\C$-valued rather than $\R$-valued, because
the imaginary part of $\kappa_{4;t_{1},t_{2},t_{3}}$ contributes
to the r.h.s.~of (\ref{eq:TrF1Fk=00003Dkappa}).

Let us employ the shorthand notation
\[
\int_{\vec{X}}^{(k)}\d\vec{\x}:=\int_{X_{1}}\d x_{1}\cdots\int_{X_{k}}\d x_{k}\int_{X_{k-1}}\d x_{k-1}'\cdots\int_{X_{2}}\d x_{2}'.
\]
Then from (\ref{eq:TrF1Fk=00003Dkappa}), we have the following formula
to express the prior conditional probability:%
\begin{align}
 & \Prob(X_{k+1},...,X_{n};t_{k+1},...,t_{n}|X_{1},...,X_{k};t_{1},...,t_{k})\nonumber \\
 & \qquad=\left(\int_{\vec{X}}^{(n)}\d\vec{\x}\ \frac{\prod_{i=1}^{n-2}\kappa_{4;t_{i},t_{i+1},t_{i+2}}(x_{i},x_{i+1},x_{i+2}',x_{i+1}')}{\prod_{i=2}^{k-2}\kappa_{2;t_{1},t_{2}}(x_{i},x_{i+1}')}\right)\nonumber \\
 & \qquad\quad\times\left(\int_{\vec{X}}^{(k)}\d\vec{\x}\ \frac{\prod_{i=1}^{k-2}\kappa_{4;t_{i},t_{i+1},t_{i+2}}(x_{i},x_{i+1},x_{i+2}',x_{i+1}')}{\prod_{i=2}^{k-2}\kappa_{2;t_{1},t_{2}}(x_{i},x_{i+1}')}\right)^{-1}.\label{eq:priorCondProb=00003Dkappa}
\end{align}
Especially, for $n=3$ and $k=2$, we find the much simpler expression:
\begin{equation}
\Prob(X_{3};t_{3}|X_{1},X_{2};t_{1},t_{2})=\frac{\int_{\vec{X}}^{(3)}\d\vec{\x}\ \kappa_{4;t_{1},t_{2},t_{3}}(x_{1},x_{2},x_{3},x_{2}')}{\int_{X_{1}}\d x_{1}\int_{X_{2}}\d x_{2}\ \kappa_{2;t_{1},t_{2}}(x_{1},x_{2})}.\label{eq:P(X3|X1X2)=00003Dkappa}
\end{equation}
For example, we can employ (\ref{eq:P(X3|X1X2)=00003Dkappa}) to describe
the result of a double-slit experiment. (Set $X_{2}=I_{1}\cup I_{2}$,
where $I_{1}$ and $I_{2}$ are disjoint intervals in $\R$.%
) For a free particle, the observable law on a double-slit experiment
ought to be Galilei invariant (or covariant). The use of (\ref{eq:P(X3|X1X2)=00003Dkappa})
makes this invariance clear, because both $\kappa_{2}$ and $\kappa_{4}$
are Galilei invariant. In this case, we find the following explicit
forms of $\kappa_{2}$ and $\kappa_{4}$ by (\ref{eq:K-free}) and
straightforward calculations.
\begin{align*}
\kappa_{2;t_{1},t_{2}}(x_{1},x_{2}) & =\frac{m}{2\pi\hbar T_{21}},\\
\kappa_{4;t_{1},t_{2},t_{3}}(x_{1},x_{2},x_{3},x_{2}') & =\frac{m^{2}}{(2\pi\hbar)^{2}T_{21}T_{32}}\exp\left(\frac{2\im m}{\hbar T_{21}T_{32}T_{31}}\left[\left(S(A_{1},A_{2},A_{3})\right)^{2}-\left(S(A_{1},A_{2}',A_{3})\right)^{2}\right]\right),
\end{align*}
where $T_{kl}:=t_{k}-t_{l}$, $A_{i}:=(x_{i},t_{i})\in\R^{2}$, $i=1,2,3$
with $A_{2}':=(x_{2}',t_{2})$, and $S(A_{1},A_{2},A_{3})$ denotes
the area of the triangle $A_{1}A_{2}A_{3}$. See Yamashita \cite{Yam2025}.

Notice that the formula (\ref{eq:priorCondProb=00003Dkappa}) refers
neither to any vector in a Hilbert space $\cH$ (usually interpreted
as a pure state), nor to any operator on $\cH$, including any state
operator (aka density matrix). Thus (\ref{eq:priorCondProb=00003Dkappa})
can be seen as an example of a description of an empirical law in
``quantum mechanics without states''.

\subsection{The POVM picture of QM}

\label{subsec:The-POVM-picture}

For each time $t$, I considered only a single quantity $\quantity_{t}$,
the position of a particle at $t$, and the corresponding projection-valued
measure (PVM) $E_{t}(\cdot)$, in Subsection \ref{subsec:A-spinless-particle}.
In this subsection, instead I consider a single positive-operator-valued
measure (POVM) for each time. A single PVM generates a commutative
$W^{*}$-algebra, where any quantum (non-commutative) aspects are
ignored. Contrary, a single noncommutative POVM has a rich quantum
implications. The spin (and more generally the angular momentum) of
a particle is a typical example whose non-commutative aspects should
be understood even when we ignore the time evolution. It turns out
that the quantum spin is suitable for the POVM formalism. Generally,
many quantum systems concerning the square-integrable irreducible
unitary representations of a Lie group appear to be suitable for the
POVM formalism. (Note that any irreducible unitary representation
of a compact Lie group (e.g., $SO(3)$, $SU(2)$ or ${\rm Spin}(n)$)
is square-integrable.) However, I shall not go into the matter of
unitary representations in this paper.

Logically, the notion of POVM is a generalization of that of PVM.
However, in this subsection I consider the POVM's which are not PVM's,
and hence the content of this subsection is not a generalization of
Subsection \ref{subsec:A-spinless-particle}. 

Note that the content of this subsection overlaps considerably with
Yamashita \cite[Sec.5]{Yam2026b}.

\subsubsection{POVM and Prior conditional probability}

Fix a scope $\alpha\in\cI$ and the local algebra $\fA_{\alpha}$,
and let $\Tr=\Tr_{\alpha}$ denote the canonical trace on $\fA_{\alpha}$.
Assume that $\fA_{\alpha}$ is a $\sigma$-finite type I factor, that
is, $\fA_{\alpha}\cong\Bdd(\cH_{\alpha})$ for some separable Hilbert
space $\cH_{\alpha}$, and so identify $\fA_{\alpha}$ with $\Bdd(\cH_{\alpha})$.

Let $\SYM{\Proj(\cH_{\alpha})}{Proj(H)}:=\Proj(\fA_{\alpha})$ and
$\SYM{\AProj(\cH_{\alpha})}{AProj(H)}:=\AProj(\fA_{\alpha})$, that
is,
\[
\AProj(\cH_{\alpha})=\{|v\rangle\langle v|:\ v\in\cH_{\alpha},\,\|v\|=1\}.
\]
Consider the subset $\{\SYM{P_{x}}{Px}|x\in\Manifold\}$ of $\AProj(\cA)$,
where $\Manifold$ is an index set, which can be interpreted as a
physically meaningful parameter space, so that each $P_{x}$ is given
an empirical meaning as a yes-no type measurement. Assume that $\fA_{\alpha}$
is generated by $\{P_{x}|x\in\Manifold\}$. Practically $\Manifold$
is often assumed to be a differential manifold, since many physical
laws are expressed as differential equations. First, I only assume
that $\Manifold$ is given a $\sigma$-compact topology, i.e., $\Manifold$
is a topological space such that $\Manifold$ is the union of a countable
number of compact subsets. Further assumptions will be added later.

Let $\mu$ be a Borel measure on $\Manifold$, and assume the following.
\begin{enumerate}
\item $\mu(S)<\infty$ for all $S\in\Borel_{\cpt}(\Manifold)$ (relatively
compact Borel sets).
\item the map $\POVM:\Borel(\Manifold)\to\fA_{\alpha}$ defined by
\[
\SYM{\POVM(S)}{E(S)}\equiv\SYM{\POVM_{S}}{ES}:=\int_{S}P_{x}\,\d\mu(x),\qquad S\subset\Borel(\Manifold),
\]
is a POVM.
\item %
If $n\in\N$, $S_{1},...,S_{n}\in\Borel(\Manifold)$, and if $\mu(S_{i})>0$
for all $i=1,...,n$, then $\POVM(S_{1})\cdots\POVM(S_{n})\neq0$.
\end{enumerate}
Note that $\POVM(S)$ is of trace-class for all $S\in\Borel_{\cpt}(\Manifold)$;
In fact, $\Tr\POVM(S)=\mu(S)<\infty$. The condition (3) is not essential,
but it simplifies the description of the theory.

Although it is not very clear that each $\POVM_{S}$ has an empirical
meaning, we assume that $\POVM_{S}$ can be interpreted as a yes-no
type measurement with some error; If $S_{1},S_{2}\in\Borel(\Manifold)$
are compact, I wish to consider the prior conditional probability
$\Prob(S_{2}|S_{1})$, which is interpreted as ``the probability
of the value of the physical quantity $\Manifold$ being measured
to be in $S_{2}$, conditioned that the value of $\Manifold$ was
measured to be in $S_{1}$''. The existence of some error in the
measurement means that $\Prob(S_{2}|S_{1})<1$ can hold even if $S_{1}\subset S_{2}$,
and $\Prob(S_{2}|S_{1})>0$ can hold even if $S_{1}\cap S_{2}=\emptyset$.

Let us define the \termi{prior conditional probability} $\Prob(S_{k+1},...,S_{N}|S_{1},...,S_{k})$
for $S_{1},...,S_{n}\in\Borel_{\cpt}(\Manifold)$ as follows:
\begin{equation}
\Prob(S_{k+1},...,S_{N}|S_{1},...,S_{k})\equiv\Prob(B|A):=\frac{\Tr(AB)^{*}AB}{\Tr A^{*}A},\label{eq:def:P(S|S)}
\end{equation}
\[
A:=\POVM(S_{1})\cdots\POVM(S_{k}),\qquad B:=\POVM(S_{k+1})\cdots\POVM(S_{N}),\qquad1\le k<N,
\]
if $\Tr A^{*}A\neq0$, i.e., $A\neq0$. Here, the time evolution is
not considered yet; It shall be discussed in Subsubsec.~\ref{subsec:Time-evolution}.
We see
\[
\Prob(S_{k+1},...,S_{N}|S_{1},...,S_{k})=\frac{\Prob(S_{2},...,S_{N}|S_{1})}{\Prob(S_{2},...,S_{k}|S_{1})},
\]
and hence it suffices to consider the cases where $k=1$.%

The definition of the prior conditional probability (\ref{eq:def:P(S|S)})
will require some explanations. Although (\ref{eq:def:P(S|S)}) does
not refer to any quantum state, here I need to refer to quantum states
for the explanation in terms of the conventional QM.

Let $\rho$ denote a density matrix on $\cH_{\alpha}$. I interpret
the unnormalized state transformation $\rho\mapsto\POVM(S)\rho\POVM(S)$
as an observation (possibly with some error) that the $\Manifold$-valued
physical quantity is in $S\subset\Manifold$.%
{} However in this interpretation, the probability of the observation
is not ${\rm Tr}(\rho\POVM(S))$ but ${\rm Tr}(\rho\POVM(S)^{2})$.
Hence this probability is not additive w.r.t.~$S$ in general, that
is, even if $S_{1},S_{2}\in\Borel(\Manifold)$ and $S_{1}\cap S_{2}=\emptyset$,
generally we have
\[
\Prob_{\rho}(S_{1}\cup S_{2})\neq\Prob_{\rho}(S_{1})+\Prob_{\rho}(S_{2}),\qquad\Prob_{\rho}(S):=\Tr(\rho\POVM(S)^{2}).
\]
One may feel that this interpretation seems strange, so I will explain
this as follows. Let $S\mapsto\tilde{\POVM}(S)$ be a Naimark dilation
of $S\mapsto\POVM(S)$; that is, $\tilde{\POVM}(\cdot)$ is a PVM
on some extended Hilbert space $\cH':=\cH_{\alpha}\oplus\cK$, so
that $E_{\alpha}\tilde{\POVM}(S)\phi=\POVM(S)\phi$ for all $S\in\Borel(\Manifold)$
and $\phi\in\cH$, where $E_{\alpha}$ is the orthogonal projection
from $\cH'$ onto $\cH_{\alpha}$. If this POVM measurement is realized
as the composition of the PVM measurement $\tilde{\POVM}(S)$ and
a projective measurement $E_{\alpha}$, on $\cH$, then the resulting
joint probability is
\[
\Tr_{\cH'}(E_{\alpha}\tilde{\POVM}(S)(\rho\oplus0)\tilde{\POVM}(S)E_{\alpha})=\Tr_{\alpha}(\POVM(S)\rho\POVM(S))=\Tr_{\alpha}(\rho\POVM(S)^{2}),
\]
and hence not additive w.r.t.~$S$, but this fact causes no contradiction
or paradox. In the literature, often the unnormalized state transformation
$\rho\mapsto\POVM(S)^{1/2}\rho\POVM(S)^{1/2}$ is assumed instead
of $\rho\mapsto\POVM(S)\rho\POVM(S)$, to maintain the additivity
of the probability $\Prob_{\rho}(S):=\Tr(\rho\POVM(S))$. However,
this assumption seems to lack a physical ground, and somewhat rather
strange.

\subsubsection{Sorkin additivity and Sorkin density}

For $S\in\Borel(\Manifold)$, again let $\Prob_{\rho}(S):=\Tr(\rho\POVM(S)^{2})$.
Let $X_{k}\in\Borel(\Manifold)$ ($k=1,...,K$, $K\in\N$) and $i\neq j\then X_{i}\cap X_{j}=\emptyset$.%
{} If each $\POVM(X_{i})$ is a projection, we have the usual (classical)
additivity $\Prob_{\rho}(\bigcup_{i=1}^{K}X_{i})=\sum_{i=1}^{K}\Prob_{\rho}(X_{i})$,
and so the computation of $\Prob_{\rho}(\bigcup_{i=1}^{K}X_{i})$
is reduced to that of each $\Prob_{\rho}(X_{i})$. However, if $\POVM(X_{i})$'s
are not projections, how can we reduce the computation of $\Prob_{\rho}(\bigcup_{i=1}^{K}X_{i})$?
The answer is given by (\ref{eq:Sorkin-S}) below. First, let $\SYM{\Bdd_{1}(\cH)}{B1()}$
denote the set of bounded operators $A$ on $\cH$ such that $\|A\|\le1$.
Then we find the following.
\begin{lem}
[Sorkin additivity]\label{lem:Sorkin0} Let $\cH$ be a (finite or
infinite-dimensional) Hilbert space. For any $A\in\Bdd_{1}(\cH)$,
set $\SYM{\Prob_{\rho}(A)}{Prho(A)}:=\Tr A\rho A^{*}$. Let $K\in\N$,
and $A_{i}\in\Bdd_{1}(\cH)$ ($i=1,...,K$). Assume $A_{i}+A_{j}\in\Bdd_{1}(\cH)$
($i\neq j$), and $A_{1}+\cdots+A_{K}\in\Bdd_{1}(\cH)$. Then we have%
\begin{equation}
\Prob_{\rho}\left(A_{1}+\cdots+A_{K}\right)=\sum_{1\le i<j\le K}\Prob_{\rho}(A_{i}+A_{j})-(K-2)\sum_{i=1}^{K}\Prob_{\rho}(A_{i}).\label{eq:Sorkin-A}
\end{equation}
\end{lem}

While the proof is easy, it seems that the significance of (\ref{eq:Sorkin-A})
as a quantum-probabilistic law is emphasized for the first time in
1994 by Sorkin \cite{Sor94} (see also \cite{Tab2009,Bar-Mul-Udu-2014,Mue2021}).
Thus we refer the law (\ref{eq:Sorkin-A}) as the \termi{Sorkin additivity}.
Let $A_{k}:=\POVM(X_{k})$, $k=1,...,K$, then we see that they satisfy
the assumptions of Lemma \ref{lem:Sorkin0}. Hence we have the following
quantum-probabilistic law in our setting:
\begin{equation}
\Prob_{0}\left(\bigcup_{i=1}^{K}X_{i}\right)=\sum_{1\le i<j\le K}\Prob_{0}(X_{i}\cup X_{j})-(K-2)\sum_{i=1}^{K}\Prob_{0}(X_{i}),\label{eq:Sorkin-S}
\end{equation}
where $\SYM{\Prob_{0}(X)}{P0()}:=\Prob(X|S_{0}),\ S_{0}\in\Borel_{\cpt}(\Manifold),\ \mu(S_{0})>0.$%
{} More generally, we have
\begin{cor}
We write $\Prob_{A}(B):=\Prob(B|A)$ in (\ref{eq:def:P(S|S)}). Let
$S_{j}\in\Borel(\Manifold)$ for $j=1,...,k$, and $S_{j}^{(i)}\in\Borel(\Manifold)$
for $i=1,...,K$, $j=k+1,...,N$. Let $A:=\POVM(S_{1})\cdots\POVM(S_{k}),$
and $B_{i}:=\POVM(S_{k+1}^{(i)})\cdots\POVM(S_{N}^{(i)})$. Assume
$B_{i}+B_{j}\in{}_{1}(\cH)$ ($i\neq j$), and $B_{1}+\cdots+B_{K}\in{}_{1}(\cH)$.
Then
\begin{equation}
\Prob_{A}\left(\sum_{i=1}^{K}B_{i}\right)=\sum_{1\le i<j\le K}\Prob_{A}(B_{i}+B_{j})-(K-2)\sum_{i=1}^{K}\Prob_{A}(B_{i}).\label{eq:def:Sorkin-B}
\end{equation}
\end{cor}

However the above Corollary is not very convenient. %
{} Instead we want the ``integral version'' of the Sorkin additivity,
where the summations are replaced by integrals. This is given by the
Sorkin density functions in Proposition \ref{prop:SorkinDensity}
below. %
First notice the following relation, which is analogous to (\ref{eq:Sorkin-A}):
\begin{lem}
Let $\mu$ be a measure on a set $X$. Let $f:X\to\C$ be measurable.
If $S\subset X$ is measurable and $\mu(S)<\infty$, we have 
\begin{align*}
\left|\int_{S}f(x)\d\mu(x)\right|^{2} & =\frac{1}{2}\int_{S}\d\mu(x)\int_{S}\d\mu(y)\left|f(x)+f(y)\right|^{2}-\mu(S)\int_{S}\left|f(x)\right|^{2}\d\mu(x)\\
 & =\frac{1}{2}\int_{S}\d\mu(x)\int_{S}\d\mu(y)\left(\left|f(x)+f(y)\right|^{2}-\left|f(x)\right|^{2}-\left|f(y)\right|^{2}\right).
\end{align*}
\end{lem}

\begin{prop}
[Sorkin density function]\label{prop:SorkinDensity} %
Let $j\in\N$, and
\[
\vec{x}^{(j)}=(x_{1}^{(j)},...,x_{j}^{(j)})\in\Manifold^{j},\qquad\vec{X}^{(j)}=(x_{0}^{(j)},...,x_{j+1}^{(j)})\in\Manifold^{j+2},\qquad j\in\N,
\]
and similarly for $\vec{y}^{(j)}$ and $\vec{Y}^{(j)}$, but let $y_{0}^{(j)}:=x_{0}^{(j)}$
and $y_{j+1}^{(j)}:=x_{j+1}^{(j)}$. We write briefly
\[
\int_{\vec{S}}\d\vec{x}^{(j)}:=\int_{S_{1}}\d\mu(x_{1}^{(j)})\cdots\int_{S_{j}}\d\mu(x_{j}^{(j)}).
\]
Then there exist two functions $\rho_{1}^{(j+2)}:\Manifold^{j+2}\to[0,\infty)$
and $\rho_{2}^{(j+2)}:\Manifold^{j+2}\times\Manifold^{j+2}\to[0,\infty)$,
called \termi{1-path and 2-path Sorkin density functions} respectively,
such that for all $S_{i}\in\Borel_{\cpt}(\Manifold)$, $i=1,...,N$,
\begin{align*}
\Prob(S_{2},...,S_{j}|S_{1}) & =C_{1}^{-1}\int_{\Manifold}\d\mu(x_{0}^{(j)})\int_{\Manifold}\d\mu(x_{j+1}^{(j)})\\
 & \quad\times\left[\int_{\vec{S}}\d\vec{x}^{(j)}\int_{\vec{S}}\d\vec{y}^{(j)}\,\rho_{2}^{(j)}(\vec{X}^{(j)},\vec{Y}^{(j)})-\mu(\vec{S})\int_{\vec{S}}\d\vec{x}^{(j)}\,\rho_{1}^{(j)}(\vec{X}^{(j)})\right],
\end{align*}
where $\mu(\vec{S}):=\prod_{i=1}^{j}\mu(S_{i})$ and $\SYM{C_{1}}{C1}:=\int_{S_{1}}\d\mu(x)\int_{S_{1}}\d\mu(y)\,\rho_{1}^{(2)}(x,y)$;
Or equivalently
\begin{align}
\Prob(S_{2},...,S_{j}|S_{1}) & =C_{1}^{-1}\int_{\Manifold}\d\mu(x_{0}^{(j)})\int_{\Manifold}\d\mu(x_{j+1}^{(j)})\int_{\vec{S}}\d\vec{x}^{(j)}\int_{\vec{S}}\d\vec{y}^{(j)}\nonumber \\
 & \qquad\qquad\times\left(\rho_{2}^{(j)}(\vec{X}^{(j)},\vec{Y}^{(j)})-\rho_{1}^{(j)}(\vec{X}^{(j)})-\rho_{1}^{(j)}(\vec{Y}^{(j)})\right).\label{eq:SorkinDensity}
\end{align}
In fact, $\rho_{1}^{(j+2)}$ and $\rho_{2}^{(j+2)}$ are given as
follows:
\begin{align}
\SYM{\rho_{1}^{(j+2)}(\vec{X}^{(j)})}{rho1} & :=\Tr A(\vec{X}^{(j)})A(\vec{X}^{(j)})^{*},\label{eq:def:rho1}\\
\SYM{\rho_{2}^{(j+2)}(\vec{X}^{(j)},\vec{Y}^{(j)})}{rho2} & :=\Tr\left(A(\vec{X}^{(j)})+A(\vec{Y}^{(j)})\right)\left(A(\vec{X}^{(j)})+A(\vec{Y}^{(j)})\right)^{*},\label{eq:def:rho2}
\end{align}
where $A(\vec{X}^{(j)}):=P_{x_{0}}\cdots P_{x_{j+1}}$.
\end{prop}

\begin{proof}
Fix $j\in\N$, and omit ``$(j)$'' in $\vec{x}^{(j)}$, $\vec{X}^{(j)}$,
etc., and ``$(j+2)$'' in $\rho_{1}^{(j+2)}$ and $\rho_{2}^{(j+2)}$.
For each $x\in\Manifold$, let $v_{x}\in\cH_{\alpha}$ be such that
$P_{x}=|v_{x}\rangle\langle v_{x}|$. %
Let $F(\vec{X}):=\prod_{i=0}^{N}\langle v_{x_{i}}|v_{x_{i+1}}\rangle$.
{} Then we find
\[
\rho_{2}(\vec{X},\vec{Y})-\rho_{1}(\vec{X})-\rho_{1}(\vec{Y})=\Tr\left(A(\vec{X})A(\vec{Y})^{*}+A(\vec{Y})A(\vec{X})^{*}\right)=2\Re\left(F(\vec{X})\ol{F(\vec{Y})}\right),
\]
and hence
\begin{align*}
 & \Tr\POVM(S_{1})\cdots\POVM(S_{j})\POVM(S_{j})\cdots\POVM(S_{1})\\
 & \qquad=\int_{\Manifold}\d\mu(x_{0})\int_{\Manifold}\d\mu(x_{j+1})\ \Tr P_{x_{0}}\POVM(S_{1})\cdots\POVM(S_{j})P_{x_{j+1}}\POVM(S_{j})\cdots\POVM(S_{1})\\
 & \qquad=\int_{\Manifold}\d\mu(x_{0})\int_{\Manifold}\d\mu(x_{j+1})\int_{\vec{S}}\d\vec{x}\int_{\vec{S}}\d\vec{y}\ F(\vec{X})\ol{F(\vec{Y})}\\
 & \qquad=\int_{\Manifold}\d\mu(x_{0})\int_{\Manifold}\d\mu(x_{j+1})\int_{\vec{S}}\d\vec{x}\int_{\vec{S}}\d\vec{y}\ \Re\left(F(\vec{X})\ol{F(\vec{Y})}\right)\\
 & \qquad=\frac{1}{2}\int_{\Manifold}\d\mu(x_{0})\int_{\Manifold}\d\mu(x_{j+1})\int_{\vec{S}}\d\vec{x}\int_{\vec{S}}\d\vec{y}\ \left(\rho_{2}(\vec{X},\vec{Y})-\rho_{1}(\vec{X})-\rho_{1}(\vec{Y})\right).
\end{align*}
Thus (\ref{eq:SorkinDensity}) follows.%
\end{proof}

\subsubsection{Three-point function}

A explicit formula to calculate the prior conditional probability
has given by (\ref{eq:SorkinDensity}), in terms of Sorkin density
functions. Next we consider another formula (\ref{eq:priorCondProd=00003Dzeta})
with (\ref{eq:zetan}) below; We will see that its empirical significance
is less clear, but it is simpler and more convenient for practical
calculations, than (\ref{eq:SorkinDensity}). For $n\in\N$, define
the function $\zeta_{n}:\Manifold^{n}\to\C$ by
\[
\zeta_{n}(x_{1},...,x_{n}):=\Tr P_{x_{1}}\cdots P_{x_{2}},\qquad x_{1},...,x_{n}\in\Manifold.
\]
Then (\ref{eq:priorCondProd=00003Dzeta}) and (\ref{eq:zetan}) say
that it is sufficient to know the ``\termi{three-point function}''
$\zeta_{3}$ to calculate all the prior conditional probabilities
of the form $\Prob(S_{2},...,S_{N}|S_{1})$. Note that $\zeta_{2}(x_{1},x_{2})=\zeta_{3}(x_{1},x_{2},x_{2})$,
and so $\zeta_{2}$ is derived from $\zeta_{3}$. We see $\zeta_{2}(x_{1},x_{2})=\Prob(P_{x_{2}}|P_{x_{1}})$,
and hence $\zeta_{2}(x_{1},x_{2})$ is an observable quantity as a
probability. Contrary, $\zeta_{3}$ is only partially observable.
For example, the absolute value%
\[
|\zeta_{n}(x_{1},...,x_{n})|=\prod_{i=1}^{n}\sqrt{\zeta_{2}(x_{i},x_{i+1})},\qquad x_{n+1}:=x_{1},
\]
is an observable quantity for all $n\ge2$. However, the complex conjugate
$\zeta_{3}'(x_{1},x_{2},x_{2}):=\ol{\zeta_{3}(x_{1},x_{2},x_{2})}$
gives the same prior conditional probabilities, and hence the two
three-point functions $\zeta_{3}$ and $\zeta_{3}'$ cannot be distinguished
experimentally.

By (\ref{eq:def:rho1}), the 1-path Sorkin density function $\rho_{1}^{(n)}$
is expressed by $\zeta_{2}$ as
\[
\rho_{1}^{(n)}(x_{1},...,x_{n})=\zeta_{2n}(x_{1},...,x_{n},x_{n},...,x_{1})=\prod_{i=1}^{n}\zeta_{2}(x_{i},x_{i+1}),\qquad x_{n+1}:=x_{1}.
\]
By (\ref{eq:def:rho2}), the 2-path Sorkin density function $\rho_{2}^{(n)}(\vec{x},\vec{y})$,
$\vec{x}=(x_{1},...,x_{n})$, $\vec{y}=(y_{1},...,y_{n})$, $y_{1}=x_{1}$,
$y_{n}=x_{n}$, is expressed by $\zeta_{2n-2}$ as
\begin{align*}
\rho_{2}^{(n)}(\vec{x},\vec{y}) & =\Tr(X+Y)^{*}(X+Y),\qquad X:=P_{x_{1}}\cdots P_{x_{n}},\ Y:=P_{y_{1}}\cdots P_{y_{n}}\\
 & =\Tr X^{*}X+\Tr Y^{*}Y+2\Re X^{*}Y\\
 & =\rho_{1}^{(n)}(\vec{x})+\rho_{1}^{(n)}(\vec{y})+2\Re\zeta_{2n}(x_{n},...,x_{1},y_{1},...,y_{n})\\
 & =\rho_{1}^{(n)}(\vec{x})+\rho_{1}^{(n)}(\vec{y})+2\Re\zeta_{2n-2}(x_{n},...,x_{1},y_{2},...,y_{n-1}).
\end{align*}
Thus the real part of $\zeta_{j}(x_{1},...,x_{j})$ ($j\ge2$) is
considered to be an observable quantity. Since both the absolute value
$|\zeta_{j}|$ and the real part $\Re\zeta_{j}$ are observable, the
value of $|\Im\zeta_{j}|$ is also observable. In other words, the
only ``empirical indefiniteness'' of $\zeta_{j}$ is in ${\rm sgn}(\Im\zeta_{j})$,
and so the only ``gauge transformation'' is the complex conjugation
$\zeta_{j}\mapsto\ol{\zeta_{j}}$.

With the shorthand notations
\begin{align*}
\vec{\x}^{(j)} & :=(\x_{1},...,\x_{2j})=(x_{1},...,x_{j},x_{j}',...,x_{1}'),\\
\int_{\vec{S}}\d\vec{\x}^{(j)} & :=\int_{S_{1}}\d\mu(x_{1})\cdots\int_{S_{j}}\d\mu(x_{j})\int_{S_{j}}\d\mu(x_{j}')\cdots\int_{S_{1}}\d\mu(x_{1}').
\end{align*}
we find the following expression:
\begin{equation}
\Prob(S_{2},...,S_{N}|S_{1})=C_{1}^{-1}\int_{\vec{S}}\d\vec{\x}^{(j)}\ \zeta_{2N}(\vec{\x}^{(j)}),\label{eq:priorCondProd=00003Dzeta}
\end{equation}
where
\[
C_{1}:=\Tr\POVM(S_{1})^{2}=\int_{S_{1}}\d\mu(x)\int_{S_{1}}\d\mu(y)\ \zeta_{2}(x,y).
\]

\begin{lem}
Let $x_{1},...,x_{n}\in\Manifold$, $n\ge4$. Then we have
\begin{equation}
\zeta_{n}(x_{1},...,x_{n})=\frac{\prod_{i=2}^{n-1}\zeta_{3}(x_{1},x_{i},x_{i+1})}{\prod_{i=3}^{n-1}\zeta_{2}(x_{1},x_{i})},\label{eq:zetan}
\end{equation}
if the denominator is nonzero.
\end{lem}

\begin{rem}
In most cases, $\zeta_{2}(x_{1},x_{2})>0$ holds for almost all $(x_{1},x_{2})\in\Manifold^{2}$,
and hence $\zeta_{n}(x_{1},...,x_{n})$ can be calculated by (\ref{eq:zetan})
for almost all $(x_{1},...,x_{n})\in\Manifold^{n}$.
\end{rem}

\begin{proof}
For each $x\in\Manifold$, let $v_{x}\in\cH_{\alpha}$ be such that
$P_{x}=|v_{x}\rangle\langle v_{x}|$. Let $x_{1},...,x_{n}\in\Manifold$%
, and $z_{kl}:=\langle v_{x_{k}}|v_{x_{l}}\rangle$, $k,l=1,...,n$.
We see 
\[
\zeta_{n}(x_{1},...,x_{n})=\prod_{i=1}^{n}z_{i,i+1},\qquad z_{n,n+1}:=z_{n,1},\qquad\text{especially }\zeta_{2}(x_{1},x_{2})=|z_{12}|^{2}.
\]
Thus we have
\begin{align*}
\prod_{i=2}^{n-1}\zeta_{3}(x_{1},x_{i},x_{i+1}) & =\prod_{i=2}^{n-1}z_{1,i}z_{i,i+1}z_{i+1,1}\\
 & =\left(\prod_{i=3}^{n-1}\ol{z_{1,i}}\right)\left(\prod_{i=3}^{n-1}z_{1,i}\right)\left(\prod_{i=1}^{n}z_{i,i+1}\right)\\
 & =\left(\prod_{i=3}^{n-1}\zeta_{2}(x_{1},x_{i})\right)\zeta_{n}(x_{1},...,x_{n}).
\end{align*}
Thus (\ref{eq:zetan}) follows.%
\end{proof}

\subsubsection{K\"{a}hler manifold and holonomy}

\global\long\def\projspace{\scP}%

In many practical cases, $\Manifold$ is a differential manifold,
and further is a K\"{a}hler manifold.\footnote{Usually, the term ``K\"{a}hler manifold'' refers to a manifold which
is both complex and Riemannian, and further those two structures are
``compatible''. Although Woodhouse \cite[p.93]{Woo92} uses the
term to refer to a compatibly complex and \emph{pseudo}-Riemannian
manifold, in this paper the term is used in the usual sense.} If $\cH_{\alpha}\cong\C^{d}$, $d\in\N$, the space $\AProj(\cH_{\alpha})$
can be viewed as the%
{} projective space $\bP^{d-1}\C$, and especially as a complex $(d-1)$-dimensional
compact K\"{a}hler manifold. Even if $\dim\cH_{\alpha}=\infty$, $\AProj(\cH_{\alpha})$
can be viewed a kind of infinite-dimensional K\"{a}hler manifold.
If $\Manifold$ is a (finite-dimensional) K\"{a}hler manifold, the
map $\Manifold\to\AProj(\cH_{\alpha})$, $x\mapsto P_{x}$ is assumed
to be a (regular) embedding, so that $\Manifold$ can be identified
with the submanifold $P_{\Manifold}:=\{P_{x}|x\in\Manifold\}$ of
$\AProj(\cH_{\alpha})$ as K\"{a}hler manifolds.

If $\SYM{\projspace_{\alpha}}{Palp}:=\AProj(\cH_{\alpha})$ is seen
as a possibly infinite-dimensional K\"{a}hler manifold, the tangent
space $T_{P}\projspace_{\alpha}$ of $\projspace_{\alpha}$ at $P\in\projspace_{\alpha}$
is canonically identified with the linear subspace
\[
\SYM{\cT_{P}}{TP}:=\{|u\rangle\langle v|:v\in\ran(P),\ u\in\ker(P)\}
\]
of $\Bdd(\cH_{\alpha})$%
, which is given an inner product $\langle X|Y\rangle=\Tr X^{*}Y$,
$X,Y\in\cT_{P}$, together with the Riemannian metric (real inner
product) $(\cdot|\cdot):=\Re\langle\cdot|\cdot\rangle$ and the symplectic
form $\omega(\cdot,\cdot):=\Im\langle\cdot|\cdot\rangle$. (Precisely,
if $P,P'\in\projspace_{\alpha}$ and $P\neq P'$, the intersection
of two tangent spaces $T_{P}\projspace_{\alpha}\cap T_{P'}\projspace_{\alpha}$
should be empty, whereas $\cT_{P}\cap\cT_{P'}=\{0\}$. Thus the above
$|u\rangle\langle v|$ should be replaced by $(P,|u\rangle\langle v|)$.)

Let $\SYM{\bbS(\cH_{\alpha})}{S(H)}:=\{v\in\cH_{\alpha}:\|v\|=1\}$.
Let $\curve:[0,1]\to\projspace_{\alpha}$ (resp.~$c:[0,1]\to\bbS(\cH_{\alpha})$)
be a piecewise smooth curve on $\projspace_{\alpha}$ (resp.~$\bbS(\cH_{\alpha})$).%
{} The curve $c$ is called a \termi{horizontal lift} of $\curve$
(or a parallel transport along $\curve$) if $c(t)\in\ran(\curve(t))$
and $\langle c(t)|\dot{c}(t)\rangle=0$ for $t\in[0,1]$. If $\curve$
is closed, i.e., $\curve(0)=\curve(1)$, and $c$ is a horizontal
lift of $\curve$, $\SYM{\hol(\curve)}{hol}:=\langle c(0)|c(1)\rangle\in U(1)$
is called the holonomy along $\curve$. 

Let $u,v\in\bbS(\cH_{\alpha})$, $\langle u|v\rangle=0$. Then the
curve $c(t):=v\cos t+u\sin t$ ($t\in\R$) is a horizontal lift of
$\curve(t):=|c(t)\rangle\langle c(t)|$. In this case, let us call
$\curve(t)$ a \termi{great circle} in $\projspace_{\alpha}$. Any
great circle is a geodesic line in $\projspace_{\alpha}$. For $P\in\projspace_{\alpha}$
and $X=|u\rangle\langle v|\in\cT_{P}$ ($v\in\ran(P),\ u\in\ker(P)$,
$\|u\|=\|v\|=1$), this great circle is expressed by
\[
\curve(t)=U(t)E_{v}U(-t),\qquad E_{v}:=|v\rangle\langle v|,\quad U(t):=\exp(t(X-X^{*})).
\]

If $\Manifold$ is identified with $P_{\Manifold}:=\{P_{x}|x\in\Manifold\}$,
then $\Manifold$ becomes a K\"{a}hler submanifold of $\projspace_{\alpha}$.
For a piecewise smooth closed curve $\curve:[0,1]\to\Manifold$, $\hol(\curve)$
is defined to be the holonomy along the closed curve on $\projspace_{\alpha}$. 

Next, let us examine the empirical significance of the above geometric
notion of holonomy, considering a conceptual double-slit experiment.

\paragraph{Interference.}

Assume that the (experimental) operation of a half mirror is simply
expressed by $L_{1/2}:=2^{-1/2}\bbOne\in\fA_{\alpha}$. Let $U,V\in\fA_{\alpha}$
satisfy $\|U\|,\|V\|\le1$, and interpret them as two operations.
Then the probability
\[
\Prob(B|A)=\frac{\Tr(AB)^{*}AB}{\Tr A^{*}A},\qquad A:=\POVM(S_{1}),\quad B:=L_{1/2}\left(U+V\right)L_{1/2}\POVM(S_{2})=\frac{1}{2}\left(U+V\right)\POVM(S_{2})
\]
can be understood to be a description of the interference between
the two operations $U$ and $V$. %

Let $a_{i},b_{i}\in\Manifold$, $i=1,...,n$, and assume $a_{1}=b_{1}$
and $a_{n}=b_{n}$. Roughly speaking, I wish to consider the experiment
of the interference between the two paths $a_{1}\to\cdots\to a_{n}$
and $b_{1}\to\cdots\to b_{n}$. Note that the K\"{a}hler manifold
$\Manifold$ is given a canonical metric $d$, determined by its Riemannian
metric. For $x\in\Manifold$ and $\epsilon>0$, let $\SYM{B_{\epsilon}(x)}{Bep()}\subset\Manifold$
denote the $\epsilon$-ball in $\Manifold$ centered at $x$, w.r.t.~the
metric $d$. For simplicity, assume that the volume $\mu(B_{\epsilon}(x))$
of the ball $B_{\epsilon}(x)$ depends on $\epsilon$ only, and $\mu(B_{\epsilon}(x))\approx C\epsilon^{N}$
for $\epsilon\approx0$, $N:=\dim\Manifold$, $C>0$. Let $S_{i}:=B_{\epsilon}(a_{i}),\ T_{i}:=B_{\epsilon}(b_{i})$,
$i=1,...,n$ (and so $S_{1}=T_{1}$, $S_{n}=T_{n}$). The interference
between the two paths will be described by the following probability:
\begin{align}
 & \Prob(B|A):=\frac{\Tr(AB)^{*}(AB)}{\Tr A^{*}A},\qquad A:=\POVM(S_{1}),\quad B:=\frac{1}{2}\left(\POVM(S_{2})\cdots\POVM(S_{n})+\POVM(T_{2})\cdots\POVM(T_{n})\right).\label{eq:P(M(S)+M(T))}
\end{align}
We see
\begin{align*}
\Prob(B|A) & =\frac{1}{2}C_{1}^{-1}\Tr\left(X+Y\right)^{*}\left(X+Y\right),\qquad X:=\POVM(S_{1})\cdots\POVM(S_{n}),\ Y:=\POVM(T_{1})\cdots\POVM(T_{n})\\
 & =\frac{1}{2}C_{1}^{-1}\left[\Tr X^{*}X+\Tr Y^{*}Y+2\Re\Tr X^{*}Y\right]\\
 & =\frac{1}{2}\left[\Prob(S_{2},...,S_{n}|S_{1})+\Prob(T_{2},...,T_{n}|T_{1})+2C_{1}^{-1}\Re\Tr X^{*}Y\right].
\end{align*}
and
\[
\Tr X^{*}Y=\int_{S_{1}}\d x_{1}\cdots\int_{S_{n}}\d x_{n}\int_{T_{1}}\d y_{1}\cdots\int_{T_{n}}\d y_{n}\ \zeta_{2n}(x_{n},...,x_{1},y_{1},...,y_{n}),
\]
but when $\epsilon\approx0$, this is approximated by $\zeta_{2n}$
as
\[
\Tr X^{*}Y\approx\left(C\epsilon^{N}\right)^{2n}\zeta_{2n}(a_{n},...,a_{1},b_{1},...,b_{n}).
\]
Alternatively, we have the following more direct approximation by
the 2-path Sorkin density function $\rho_{2}^{(n)}$:
\[
\Prob(B|A)=\frac{1}{2}C_{1}^{-1}\Tr\left(X+Y\right)^{*}\left(X+Y\right)\approx\frac{1}{2}\left(C\epsilon^{N}\right)^{2n}C_{1}^{-1}\rho_{2}^{(n)}((a_{1},...,a_{n}),(b_{1},...,b_{n}))
\]

\paragraph{Interference between two continuous paths.}

Next, consider a more idealized and simplified interference experiment.
Let $\curve:[0,1]\to\Manifold$ be a piecewise smooth curve. Let
\[
\SYM{A(\curve,n)}{A(ga,L)}:=P_{a_{0}}\cdots P_{a_{n}},\qquad a_{i}:=\curve(i/n),\quad i=0,...,n,\qquad n\in\N,
\]
and
\begin{equation}
\SYM{A(\curve)}{A(gam)}:=\lim_{n\to\infty}A(\curve,n).\label{eq:def:A(C)}
\end{equation}
{} %
We see $A(\curve)A(\curve)^{*}=P_{\curve(0)}$ and $A(\curve)^{*}A(\curve)=P_{\curve(1)}$.
If $\curve_{1},\curve_{2}:[0,1]\to\Manifold$ are two piecewise smooth
curves such that $\curve_{1}(0)=\curve_{2}(0)$, $\curve_{1}(1)=\curve_{2}(1)$.
Then $A(\curve_{1})=e^{\im t}A(\curve_{2})$ for some $t\in\R$. %
{} %
When $P_{\curve(0)}P_{\curve(1)}\neq0$ , let
\[
T_{\curve}:=\frac{P_{\curve(0)}P_{\curve(1)}}{\sqrt{\Tr P_{\curve(0)}P_{\curve(1)}}},
\]
then we see $A(\curve,L)=zT_{\curve}$ for some $z\in\C,\ |z|\le1$,
and $A(\curve)=zT_{\curve}$ for some $z\in U(1)$, i.e., $z\in\C,\ |z|=1.$

If $\curve(0)=\curve(1)$, we find
\begin{equation}
A(\curve)=\hol(\curve)^{-1}P_{\curve(0)}.\label{eq:A(C)=00003Dhol}
\end{equation}
Let us check this when $\curve$ is a small geodesic triangle in $\projspace_{\alpha}:=\AProj(\cH_{\alpha})$.
(Although a geodesic line on $\Manifold$ is not a geodesic line on
$\projspace_{\alpha}$ in general (and vice versa), any \emph{small}
geodesic line segment on $\Manifold$ will be well approximated by
a geodesic line segment on $\projspace_{\alpha}$. Furthermore, recall
that when the holonomy group is $U(1)$, the full information of the
holonomy $\hol(\cdot)$ can be reconstructed from the ``infinitesimal
holonomy'', that is, $\hol(\curve)$ for the infinitesimal loops
$\curve$, in the following sense. When the holonomy group is $U(1)$,
we can define the holonomy $\hol(\curve)$ of the oriented but non-parameterized
loop $\curve$ (i.e., $\curve$ is an oriented submanifold of $\Manifold$
which is diffeomorphic to $S^{1}$)%
, and furthermore, if a bounded surface $D\subset\Manifold$ is decomposed
as $D=D_{1}\cup\cdots\cup D_{n}$, then we have $\hol(\di D)=\prod_{i=1}^{n}\hol(\di D_{i})$,
where $\di$ denotes the boundary.)

Let $u,v\in\bbS(\cH_{\alpha})$, $\langle u|v\rangle=0$, and consider
the segment of a great circle $\curve(t):=E_{t}:=|c(t)\rangle\langle c(t)|$,
$c(t):=v\cos\epsilon t+u\sin\epsilon t$ ($t\in[0,1]$), where $\epsilon>0$
is a small constant.%

Let $t_{i}:=i/n$, $i=0,...,n$, then we have
\[
A(\curve,n)=E_{t_{0}}\cdots E_{t_{n}}=\left(\prod_{i=0}^{n-1}\cos\epsilon(t_{i}-t_{i+1})\right)|c(0)\rangle\langle c(1)|=\left(\cos\frac{\epsilon}{n}\right)^{n}|c(0)\rangle\langle c(1)|,
\]
and hence
\begin{equation}
A(\curve)=\lim_{n\to\infty}A(\curve,n)=|c(0)\rangle\langle c(1)|=\frac{E_{0}E_{1}}{\langle c(0)|c(1)\rangle}=\frac{E_{0}E_{1}}{\cos\epsilon}.\label{eq:A(C)=00003DE0E1}
\end{equation}

Next, let $\curve$ be a small geodesic triangle in $\projspace_{\alpha}$.
That is, assume that the curve $\curve$ is the concatenation of three
curves $\curve_{1},\curve_{2},\curve_{3}$, written as $\curve=\curve_{1}\bullet\curve_{2}\bullet\curve_{3}$,
and that each $\curve_{j}:[0,1]\to\projspace_{\alpha}$ ($j=1,2,3$)
is a small segment of a great circle. Of course, it is required that
$\curve_{j}(1)=\curve_{j+1}(0)$ for $j=1,2,3$ ($\curve_{4}:=\curve_{1}$).
Let
\[
\curve_{j}(t):=|c_{j}(t)\rangle\langle c_{j}(t)|,\quad c_{j}(t):=v_{j}\cos\epsilon_{j}t+u_{j}\sin\epsilon_{j}t,\quad v_{j},u_{j}\in\bbS(\cH_{\alpha}),\ v_{j}\perp u_{j}
\]
for $t\in[0,1],\ j=1,2,3.$ Furthermore assume $c_{j}(1)=c_{j+1}(0)$
for $j=1,2$, so that $c:=c_{1}\bullet c_{2}\bullet c_{3}$ is a horizontal
lift of $\curve$. Since $\langle c_{j}(0)|c_{j}(1)\rangle=\cos\epsilon_{j}>0$
for $j=1,2,3$, and $\left|\langle v_{1}|c_{3}(1)\rangle\right|=1$,
we have
\[
c_{3}(1)=\frac{|\langle v_{3}|v_{1}\rangle|}{\langle v_{3}|v_{1}\rangle}v_{1}=\frac{\cos\epsilon_{3}}{\langle v_{3}|v_{1}\rangle}v_{1},
\]
and hence $\hol(\curve)=(\cos\epsilon_{3})/\langle v_{3}|v_{1}\rangle$.

Let $F_{j}:=\curve_{j}(0)$, then from (\ref{eq:A(C)=00003DE0E1}),
we obtain (\ref{eq:A(C)=00003Dhol}) in the following form:
\[
A(\curve)=\frac{F_{1}F_{2}F_{3}F_{1}}{\cos\epsilon_{1}\cos\epsilon_{2}\cos\epsilon_{3}}=\frac{\langle v_{1}|v_{2}\rangle\langle v_{2}|v_{3}\rangle\langle v_{3}|v_{1}\rangle}{\cos\epsilon_{1}\cos\epsilon_{2}\cos\epsilon_{3}}F_{1}=\frac{\langle v_{3}|v_{1}\rangle}{\cos\epsilon_{3}}F_{1}=\hol(\curve)^{-1}F_{1}.
\]

For any curve $\curve:[0,1]\to\Manifold$, let $\curve^{-1}(t):=\curve(1-t)$.
Next consider the two curves $\curve_{1},\curve_{2}:[0,1]\to\Manifold$
satisfying $\curve_{1}(0)=\curve_{2}(0)$, $\curve_{1}(1)=\curve_{2}(1)$,
and the concatenation $\curve_{1}\bullet\curve_{2}^{-1}$. We find
that
\[
\Tr\left[A(\curve_{1})^{*}A(\curve_{2})\right]=\hol(\curve_{1}\bullet\curve_{2}^{-1}).
\]

The simplest way to see the interference between the two operations
$A(\curve_{1})$ and $A(\curve_{2})$ will be to consider the following
probability:
\[
\Prob(B|A):=\frac{\Tr(AB)^{*}AB}{\Tr A^{*}A},\quad A:=P_{a_{0}},\quad B:=\frac{1}{2}\left(A(\curve_{1})+A(\curve_{2})\right).
\]
But in this case simply we have
\begin{align*}
\Prob(B|A) & =\Tr B^{*}B=\frac{1}{4}\tilde{\rho}_{2}(\curve_{1},\curve_{2}),
\end{align*}
where $\tilde{\rho}_{2}$ is the \termi{2-path Sorkin density function for piecewise smooth paths}
defined by
\begin{align*}
\SYM{\tilde{\rho}_{2}(\curve_{1},\curve_{2})}{rho2}: & =\Tr\left(A(\curve_{1})+A(\curve_{2})\right)\left(A(\curve_{1})+A(\curve_{2})\right)^{*}\\
 & =2+2\Re\Tr A(\curve_{1})^{*}A(\curve_{2})\\
 & =2+2\Re\,\hol(\curve_{1}\bullet\curve_{2}^{-1}).
\end{align*}
Similarly, the 1-path Sorkin density function $\tilde{\rho}_{1}$
for piecewise smooth paths will be defined by $\SYM{\tilde{\rho}_{1}(\curve)}{rho1}:=\Tr A(\curve)A(\curve)^{*},$
but this is trivial: $\tilde{\rho}_{1}(\curve)\equiv1$.%

\paragraph{Path integral representation.}

If $\Manifold$ is a simply connected K\"{a}hler manifold, the 2-path
Sorkin density function $\tilde{\rho}_{2}$ is considered to determine
\emph{all} of empirical/observable laws; If a K\"{a}hler manifold
(or more generally a symplectic manifold) $\Manifold$ is simply connected,
the prequantization bundle over $\Manifold$ and the corresponding
the holonomy $\hol(\cdot)$ is unique (if it exists) \cite[p.161]{Woo92}.
Furthermore, a K\"{a}hler manifold $\Manifold$ has a canonical polarization,
that is, the holomorphic polarization \cite[p.93]{Woo92}, and hence
$\Manifold$ has a canonical quantization \cite[Sec.9.2]{Woo92}.
To summarize, the K\"{a}hler structure of $\Manifold$ uniquely determines
the ``prequantization holonomy'' $\hol(\cdot)$, and further $\hol(\cdot)$
canonically determines a quantization, where all physical laws, including
the probability $\Prob(S_{2},...,S_{N}|S_{1})$, are described.

Let $\Prob_{\hol}(S_{2},...,S_{N}|S_{1})$ denote the probability
determined by the holonomy $\hol(\cdot)$. I suggest to express $\Prob_{\hol}(S_{2},...,S_{N}|S_{1})$
by a formal ``path integral'': Consider the following formal expression,
analogous to Proposition \ref{prop:SorkinDensity}:
\begin{align}
\Prob_{\hol}(S_{2},...,S_{N}|S_{1}) & =\frac{1}{2C_{1}}\int_{\Paths(\vec{S})}\cD\curve_{1}\int_{\Paths(\vec{S})}\cD\curve_{2}\left(\rho_{2}(\curve_{1},\curve_{2})-\rho_{1}(\curve_{1})-\rho_{1}(\curve_{2})\right)\\
 & =C_{1}^{-1}\int_{\Paths(\vec{S})}\cD\curve_{1}\int_{\Paths(\vec{S})}\cD\curve_{2}\ \Re\,\hol(\curve_{1}\bullet\curve_{2}^{-1}),\label{eq:P(SS)=00003DpathInt}
\end{align}
where we assume $\curve_{1}(0)=\curve_{2}(0)$, $\curve_{1}(1)=\curve_{2}(1)$,
and 
\[
\Paths(\vec{S}):=\{\curve\in C([0,1],\Manifold)|\,\curve(j\epsilon)\in S_{j},\ j=1,...,N\},\qquad\epsilon:=\frac{1}{N+1},\qquad C_{1}:=\Tr\POVM(S_{1}).
\]

Although the formal path-integral expression (\ref{eq:P(SS)=00003DpathInt})
itself can be useful, it is more desirable for (\ref{eq:P(SS)=00003DpathInt})
to be given a rigorous measure-theoretical ground. Generally, rigorous
justifications of path integrals used in quantum physics are difficult.
However it is well-known that some rigorous formulations of path integrals
are possible in terms of \emph{Brownian motions}, mainly with the
Feynman\textendash Kac formula and its generalizations. Recall that
we can define the Brownian motion on a manifold $\Manifold$ only
if $\Manifold$ is furnished with a Riemannian metric. Since a classical-mechanical
phase space is represented as a symplectic manifold, we cannot define
the Brownian motion on a phase space in general. However sometimes
a phase space has a canonical K{\"a}hler structure, which contains
a Riemannian structure.

For the rigorous considerations for the path integrals on phase spaces,
see Daubechies and Klauder \cite{DK85}, Charles \cite{Cha1999},
and Yamashita \cite{Yam11,Yam18,Yam22a,Yam22b}. See G{\"u}neysu
\cite{Gun10,Gun17} for the generalized Feynman\textendash Kac formula
on a vector bundle over a Riemannian manifold.

\subsubsection{Time evolution}

\label{subsec:Time-evolution}

\global\long\def\K{{\bf K}}%

Let $\{U(t)|t\in\R\}$ be a one-parameter unitary group in $\fA_{\alpha}=\Bdd(\cH_{\alpha})$,
which is understood as the time evolution. Let
\[
\POVM_{t}(S):=U(t)\POVM(S)U(-t),\qquad t\in\R,\ S\in\Borel(\Manifold),
\]
and
\[
A:=\POVM_{t_{1}}(S_{1}),\quad B:=\POVM_{t_{2}}(S_{2})\cdots\POVM_{t_{n}}(S_{n}),\quad C_{1}:=\Tr\POVM_{t_{1}}(S_{1})^{2}.
\]
Consider the prior conditional probability
\begin{equation}
\Prob(S_{2},...,S_{n};t_{2},...,t_{n}|S_{1};t_{1}):=\Prob(B|A)=\frac{\Tr(AB)^{*}(AB)}{\Tr A^{*}A}=C_{1}^{-1}\Tr(AB)^{*}(AB),\label{eq:P(SSTT)-POVM}
\end{equation}
where $S_{i}\in\Borel_{\cpt}(\Manifold)$, $t_{i}\in\R$, $x_{i}\in\Manifold$,
$i=1,...,n$, and usually either $t_{1}<\cdots<t_{n}$ or $t_{1}>\cdots>t_{n}$
is assumed.

For each $x\in\Manifold$, choose $v_{x}\in\cH_{\alpha}$ such that
$P_{x}=|v_{x}\rangle\langle v_{x}|$. %
Then we find
\[
\Prob(S_{2},...,S_{n};t_{2},...,t_{n}|S_{1};t_{1})=C_{1}^{-1}\int_{S_{1}}\d\mu(x_{1})\cdots\int_{S_{N}}\d\mu(x_{N})\ \prod_{i=1}^{N}\K(x_{i},x_{i+1};t_{i},t_{i+1}),
\]
where 
\[
\K(x,x';t,t'):=\langle v_{x}|U(t'-t)v_{x'}\rangle,\quad\text{and}\quad x_{N+1}:=x_{1},\ t_{N+1}:=t_{1}.
\]
This is analogous to the propagator $K$ for a spinless particle appeared
in (\ref{eq:TrEE=00003DintintK}), and so I again call $\K$ a \termi{propagator}.
Similarly to the propagator $K$ in (\ref{eq:TrEE=00003DintintK}),
$\K$ has some ``empirically redundant information'', namely, the
choice of $v_{x}$ for each $x\in\Manifold$. %

Similarly to (\ref{eq:def:kappa}), let
\begin{equation}
\SYM{\kappa(\vec{x},\vec{t})}{kappa}\equiv\kappa_{n}(\vec{x},\vec{t}):=\prod_{i=1}^{n}\K(x_{i},t_{i};x_{i+1},t_{i+1}),\qquad\vec{x}:=(x_{1},...,x_{n}),\quad\vec{t}:=(t_{1},...,t_{n}).\label{eq:def:kappa-POVM}
\end{equation}
Then we see
\[
\kappa_{n}(\vec{x},\vec{t})=\Tr P_{x_{1},t_{1}}\cdots P_{x_{n},t_{n}},\qquad P_{x,t}:=U(t)P_{x}U(-t),
\]
and hence $\kappa_{n}$ is independent of the choice of $\{v_{x}\}_{x}$.
So it is considered to be empirically less redundant. Hence an analysis
similar to that of Subsec.~\ref{subsec:A-spinless-particle} will
be possible. However, unlike Subsec.~\ref{subsec:A-spinless-particle},
the K\"{a}hler (and so symplectic) manifold $\Manifold$ may be viewed
as a (classical-mechanical) phase space, and the function $\kappa_{n}$
can be related to a classical Hamiltonian $H:\Manifold\to\R$. For
example, we may expect that $\kappa_{n}$ (and also the prior conditional
probability (\ref{eq:P(SSTT)-POVM})) can be expressed as a path integral,
whose integrand contains the classical Hamiltonian $H$. See the literature
on geometric quantization (e.g., \cite{Woo92}), and also \cite{DK85,Cha1999,Yam11,Yam18,Yam22a,Yam22b}.

\section{QFTM and QFTCS without states}

\label{sec:QFTM-and-QFTCS}

If Working Hypothesis \ref{hyp:causal-typeIfactor} (or \ref{hyp:causal-atomic})
is assumed also in QFTM and QFTCS, it may be expected that a considerable
part of the contents of Section \ref{sec:NRQM-withoutStates} can
be also applied to them. However, if we wish to reject the use of
the concept of state (especially of vacuum state), we will need an
alternative to the Wightman formalism \cite{SW64} of QFTM, and furthermore
it should be generalizable to QFTCS. Hence this might be an extremely
hard task.

Since currently we have no example of rigorous interacting QFT in
the Minkowski spacetime $\R^{1,3}$, or in $(3+1)$-dimensional curved
spacetimes, we are obliged to guess the possible formalism of general
theory mainly from the examples of the free fields, as well as the
low dimensional nontrivial models. In this section, only the free
fields are discussed; Actually I shall examine only the general structure
of the CCR and CAR algebras, regardless of the spacetime structure.
A brief outline of the rigorous definition of the massive free scalar
(Klein\textendash Gordon) QFTCS is presented in \cite{Yam2025}, and
that of the Dirac/Majorana QFTCS is presented in \cite{Yam2026b}.
Those definitions imply that the CCR (resp.~CAR) algebra is not only
a convenient algebraic tool for the free scalar (resp.~Dirac/Majorana)
QFTCS; Rather, a free scalar (resp.~Dirac/Majorana) QFTCS is defined
as a CCR (resp.~CAR) algebra itself, in which the spacetime structure
is encoded.

\subsection{Free scalar fields}

Consider the common algebraic structure of free scalar fields in a
Minkowski spacetime $\R^{1,d}$ ($d\ge1$) or a curved spacetime (i.e.,
a globally hyperbolic Lorentzian manifold), or even in a non-relativistic
spacetime. This structure is nothing but the \emph{CCR algebra}.

This subsection is based on Yamashita \cite{Yam2025}. For the rigorous
algebraic formalism of free fields in curved spacetime, see \cite{BGP2007,BF2009,HW2015,KM2015,BD2015}.

Let $V$ be a pre-symplectic vector space, that is, a real linear
space equipped with an %
antisymmetric bilinear form $\sigma$. (A pre-symplectic vector space
$V$ is called symplectic if $\sigma$ is nondegenerate.) 
\begin{defn}
(e.g., \cite[Sec.5.2.2.2]{BR97})\label{def:CCR-C} Let $\fA$ be
a $C^{*}$-algebra generated by nonzero elements $\weyl(v),$ $v\in V$,
satisfying
\begin{enumerate}
\item $\weyl(-v)=\weyl(v)^{*},$
\item $\weyl(v)\weyl(u)=e^{-\im\sigma(v,u)/2}\weyl(v+u)$ 
\item $\weyl(v)=\bbOne$ if $\sigma(w,v)=0$ for all $w\in V$.
\end{enumerate}
for all $v,u\in V.$ It is shown that such $\fA$ is unique up to
{*}-isomorphism. So write $\fA=\SYM{\CCR(V)}{CCR()}$ and call it
the \termi{CCR $C^*$-algebra} over $V$.
\end{defn}

Let $N:=\{v\in V|\sigma(u,v)=0,\ \forall u\in V\}$. Then $V/N$ naturally
becomes a symplectic vector space. Thus it seems to suffice to consider
only the cases where $V$ is a symplectic vector space. However, for
a free scalar field, I would like to set $V=C_{\cpt}^{\infty}(\STime,\R)$,
the set of compactly supported smooth real functions on the spacetime
manifold $\STime$. In this case, the antisymmetric form $\sigma$
is defined by the causal propagator, and $\sigma$ is degenerate.
Note that the structure of the spacetime $\STime$ is encoded in the
space $C_{\cpt}^{\infty}(\STime,\R)$, not in the quotient space $C_{\cpt}^{\infty}(\STime,\R)/N$.
Thus if we set $V=C_{\cpt}^{\infty}(\STime,\R)$, it is also encoded
in the CCR algebra $\CCR(V)$.

Let $V$ be a pre-symplectic vector space, and $W\subset V$ a pre-symplectic
subspace of $V$. Then we see $\CCR(W)\subset\CCR(V)$, or more precisely,
$\CCR(W)$ is canonically embedded into $\CCR(V)$. Notice the following
fact:
\begin{prop}[{See \cite[ Proposition 5.2.9.]{BR97}}]
\label{prop:CCR(V)=00003DCCR(W)} $\CCR(V)=\CCR(W)$ if and only
if $V=W.$
\end{prop}

Note that if $V$ is finite-dimensional and $\sigma$ is nondegenerate,
the irreducible representation $(\pi,\cH)$ of $\CCR(V)$ is unique
up to unitary equivalence. Thus the notion of the \termi{CCR $W^{*}$-algebra}
{} $\SYM{\CCRW(V)}{CCRW}$ over $V$ may be defined as a canonical extension
of $\CCR(V)$, determined by $(\pi,\cH)$.%
{} However, since $\pi(\CCR(V))''=\Bdd(\cH)$, the $W^{*}$-algebraic
structure of $\pi(\CCR(V))''$ is always the (unique) $\sigma$-finite
factor of type I$_{\infty}$, independently of the dimension $\dim V\in2\N$,
despite of Proposition \ref{prop:CCR(V)=00003DCCR(W)}.%
{} Hence the CCR $W^{*}$-algebra $\CCRW(V)$ should be understood as
the ordered pair $(\fA,\weyl)$ such that
\begin{enumerate}
\item $\fA$ is a $W^{*}$-algebra (specifically a $\sigma$-finite factor
of type I$_{\infty}$),
\item $\weyl:V\to\fA$ satisfies (1)\textendash (3) in Definition \ref{def:CCR-C},
\item $\{\weyl(v)|v\in V\}$ generates $\fA$.
\end{enumerate}
Note that $\CCRW(V)$ is not defined if $V$ is infinite-dimensional.

Assume $\dim V=\infty$ and further $\dim(V/N)=\infty$. A canonical
causal net of $W^{*}$-algebras (factors of type I) is constructed
from the CCR $C^{*}$-algebra $\CCR(V)$ as follows. Let $\SYM{\Fin(V)}{Fin(V)}$
denote the set of finite-dimensional \emph{symplectic} subspaces of
$V$, that is, $Z\in\Fin(V)$ if and only if $Z$ is a finite-dimensional
subspace of $V$, and $\sigma|_{Z\times Z}$ is a symplectic (i.e.,
nondegenerate antisymmetric) form on $Z$. For each $Z\in\Fin(V)$,
consider the CCR $W^{*}$-algebra $\CCRW(Z)=(\fA_{Z},\weyl|_{Z})$,
where $\fA_{Z}$ is the $W^{*}$-algebra generated by $\CCR(Z)$.
For $Z,Z'\in\Fin(V)$, write $Z\cauindep Z'$ if $\sigma(u,v)=0$
for all $u\in Z$ and $v\in Z'$. We see that if $Z,Z'\in\Fin(V)$
and $Z\subset Z'$, then $\fA_{Z}\subset\fA_{Z'}$. Let $\fA=\fA_{V}$
be the $C^{*}$-algebra generated by $\fA_{Z}$'s, that is, $\fA$
is the completion of $\bigcup_{Z\in\Fin(V)}\fA_{Z}$ with respect
to the norm $\|\cdot\|$.

Now the following is evident (recall Definition \ref{def:causalNet}):
\begin{lem}
If $Z,Z'\in\Fin(V)$ and $Z\cauindep Z'$, then $[\fA_{Z},\fA_{Z'}]=\{0\}$
in $\fA_{V}$.
\end{lem}

\begin{lem}
$\{\fA_{Z}|Z\in\Fin(V)\}$ is a causal net indexed by the directed
set $(\Fin(V),\subset)$ with the orthogonality relation $\cauindep$.
\end{lem}

Let $\fR$ be an arbitrary $W^{*}$-algebra. Note that there exists
a one-to-one correspondence between the one-parameter unitary groups
$\{U(t)|t\in\R\}$ in $\fR$, and the PVM's $E:\Borel(\R)\to\Proj(\fR)$.
If $\fR\subset\Bdd(\cH)$, there exists a unique (possibly unbounded)
selfadjoint operator $A$ on $\cH$ such that $U(t)=e^{itA}$, and
hence the unique spectral PVM $E(\cdot)$ of $A$, and vice versa.

Assume that the CCR $C^{*}$-algebra $\CCR(V)$ and also the causal
net $\{\fA_{Z}|Z\in\Fin(V)\}$ are given a physical meaning as a free
scalar field, as follows. For each $v\in Z\in\Fin(V)$, let $E_{v}:\Borel(\R)\to\Proj(\fA_{Z})$
be the PVM corresponding to the one-parameter unitary group $\{\weyl(tv)|t\in\R\}$
in $\fA_{Z}$. Then each $E_{v}(S)$ ($S\in\Borel(\R)$) is interpreted
as a yes-no type measurement that the value of the $\R$-valued field
$\phi(v)$ is in $S$.

In this situation, I suggest that a considerable part of the physical/empirical
laws can be described on each ``local algebra'' $\fA_{Z}$, as the
prior conditional probability
\begin{equation}
\Prob_{Z}(B|A):=\frac{\Tr_{Z}(AB)^{*}AB}{\Tr_{Z}A^{*}A},\qquad A:=E_{v_{1}}(S_{1})\cdots E_{v_{k}}(S_{k}),\quad B:=E_{v_{k+1}}(S_{k+1})\cdots E_{v_{n}}(S_{n}),\label{eq:PZ}
\end{equation}
where $2\le k<n$, $v_{i}\in Z$, $S_{i}\in\Borel_{\cpt}(\R)$ ($i=1,...,n$),
and $\Tr_{Z}$ denotes the canonical trace on $\fA_{Z}$. Of course,
this probability is defined only when $0<\Tr_{Z}A^{*}A<\infty$ holds. 

For example, consider the cases where $\dim Z=2$. We find that if
$v_{1},v_{2}\in Z$ are linearly independent, and if both $S_{1}$
and $S_{2}$ are non-null (w.r.t.~the Lebesgue measure on $\R$),
then $A:=E_{v_{1}}(S_{1})E_{v_{2}}(S_{2})$ satisfy $0<\Tr_{Z}A^{*}A<\infty$.

More generally, consider the cases where $\dim Z=2N$, $N\in\N$.
If $L_{1}:=\Span\{v_{1},...,v_{N}\}\subset Z$ is a Lagrangian subspace
of $Z$, i.e., if $\sigma(u,v)=0$ for all $u,v\in L_{1}$ and $\dim L_{1}=\frac{1}{2}\dim Z\,(=N)$,
then we see that $E_{v_{1}}(S_{1}),\cdots,E_{v_{N}}(S_{N})$ are pairwise
commutative. Let $L_{2}:=\Span\{u_{1},...,u_{N}\}\subset Z$ be another
Lagrangian subspace of $Z$ such that $L_{1}\cap L_{2}=\{0\}$. Then
we find that
\[
A:=E_{v_{1}}(S_{1})\cdots E_{v_{N}}(S_{N})E_{u_{1}}(T_{1})\cdots E_{u_{N}}(T_{N}),\qquad S_{i},T_{i}\in\Borel_{\cpt}(\R),\ \text{non-null},\ i=1,...,N,
\]
satisfies $0<\Tr_{Z}A^{*}A<\infty$. See Yamashita \cite[Section 5]{Yam2025}.

However, for self-interacting scalar fields, currently I do not know
when $\Tr_{Z}A^{*}A<\infty$ holds, and even whether it can hold or
not. Indeed, it is not even clear how to construct a causal net of
type I factors (or more generally atomic $W^{*}$-algebras) for such
fields. Thus in this case my claim that a considerable part of the
empirical laws can be expressed as the prior conditional probability
of the form (\ref{eq:PZ}) is a mere guess without any strong evidence.
However, of course, currently no one knows whether there exists a
rigorous self-interacting scalar field in the Minkowski space $\R^{1,3}$,
or in a $(3+1)$-dimensional curved spacetime. (It seems likely that
the answer is no. See \cite{AD2021} and references therein.) Hence
it might not be very productive to study the empirical laws of self-interacting
scalar fields in $(3+1)$-dimension.

\subsubsection{Physical laws in a Fock space}

Consider a conventional free scalar QFTM, formulated on a boson Fock
space $\cF_{{\rm b}}(\cH)$, where $\cH$ is a Hilbert space (see
e.g., \cite[Section X.7]{RS75}). Consider the following
\begin{question}
What are the physical/empirical laws in the conventional free scalar
field theory which cannot be described in our ``stateless'' formalism
in a causal net of atomic $W^{*}$-algebras? Do such laws exist?
\end{question}

Of course, all the ``physical laws'' of the form ``if the system
is prepared in a state $\omega$, then ...'' are not described in
our stateless formalism in the literal form. However if $\omega$
refers to a local state on an atomic $W^{*}$-algebra%
, the law is rephrased as (\ref{eq:law-prepare-local-op}), employing
the notion of operation, instead of state. On the other hand, if $\omega$
refers to a global state, or a local state on a type III factor, I
doubt that this is really a physical/empirical law, as I argued in
Sections \ref{sec:Global-state} and \ref{sec:Local-state}.

In QFTM, there are several important ``global observables'' such
as the Hamiltonian $H$, and 4-momentum $(P_{0},...,P_{3})$, $P_{0}:=H$.
The operator $H$ is interpreted as the total energy or mass. For
a free scalar QFTM, we can also consider the total number of the particles,
expressed by the number operator $N$ on the Fock space $\cF_{{\rm b}}(\cH)$.
Probably, any ``physical law'' in terms of such global observables
cannot be described in our formalism. However, I again doubt that
such a law is really a physical law, since it is doubtful that a ``global
observable'' is really observable for us. Even if we can observe
such a global observable (e.g., the ``total mass'' of the Universe)
in QFTM in some sense, in a curved spacetime we cannot observe it
in general; Actually the cosmological event horizon exists in our
Universe, and we can get no information on the outside of the horizon
in principle. Probably, we will not even be able to give a mathematical
definition of the total mass in a generic curved spacetime, and so
such a notion will be almost meaningless both mathematically and empirically.

In a Fock space, we can define the Wick product of field operators.
Some important local observables such as the stress-energy tensor
are defined with the Wick products. Currently, I am not sure whether
the physical laws concerning such ``Wick-type local observables''
can be described in our framework. Notice that in QFTCS, stress-energy
tensors are defined in an algebraic formalism without Fock spaces.
For the rigorous constructions of stress-energy tensor in QFTCS, see
\cite[Subsec.3.1]{HW2015} and references therein.

\subsection{Free fermion fields}

Consider the common algebraic structure of free fermion fields in
a Minkowski spacetime or a curved spacetime, or even in a non-relativistic
spacetime. This structure is nothing but the \emph{CAR algebra}. 

This subsection is based on Yamashita \cite{Yam2026b}. For the rigorous
algebraic formalism of Dirac fields in curved spacetime, see \cite{Dim1982,DH2006,San2008,DHP2009,BD2015}. 

I use the symbol $(\cdot|\cdot)$ to denote a \emph{real} (pre-)inner
product on some (real or complex) vector space.
\begin{defn}
\label{def:CAR-real}%
{} Let $V$ be a real vector space, and $(\cdot|\cdot)$ a (real) pre-inner
product (i.e., a positive-semidefinite bilinear form) on $V$.  %
{} The \termi{CAR algebra} $\SYM{\CAR(V)}{CAR()}=\CAR(V,(\cdot|\cdot))$
over $V$ is defined to be the $C^{*}$-algebra generated by $\bbOne$
and $\phi(v),v\in V$, satisfying
\begin{enumerate}
\item $v\mapsto\phi(v)$ is $\R$-linear,
\item $\phi(v)^{*}=\phi(v),$ for all $v\in V$,
\item $\{\phi(v),\phi(u)\}=2(v|u)\bbOne,$ $v,u\in V$, where $\{X,Y\}:=XY+YX$.
\item If $(v|v)=0$, then $\phi(v)=0$.
\end{enumerate}
\end{defn}

If $\CAR(V)$ is interpreted as a fermion field, it is supposed that
$\CAR(V)$ has the $\Z_{2}$-gauge transformation $\phi(v)\mapsto-\phi(v)$,
$v\in V$. Thus any observable, which is $\Z_{2}$-gauge invariant,
belongs to the even subalgebra $\CAR(V)_{\even}$ of $\CAR(V)$. Furthermore,
assume that every selfadjoint element of $\CAR(V)_{\even}$ is an
observable, in other words, that there is no other gauge transformations.

Consider the finite-dimensional case $V=\R^{k}$, $k\in\N$. We have
\begin{align}
\CAR(\R^{2n}) & \cong\Mat(2^{n},\C),\label{eq:CAR(R2n)}
\end{align}
\[
\CAR(\R^{2n+1})\cong\Mat(2^{n},\C)\oplus\Mat(2^{n},\C)\cong\CAR(\R^{2n})\oplus\CAR(\R^{2n}),
\]
for all $n\in\N$. 

Let $\SYM{\Cl_{n}}{Cln}\equiv\Cl(\R^{n})$ denote the real Clifford
algebra over $\R^{n}$ with the usual inner product, whose generators
are $\bbOne$ and the elements of $\R^{n}$ subject to the relation
\[
v\cdot w+w\cdot v=-2(v|w),\qquad v,w\in\R^{n}.
\]
Then $\CAR(\R^{n})$ is nothing but the complexification of $\Cl_{n}$
with $\phi(v):=\im v$ and $v^{*}:=-v$. By \cite[Chapter 1, Theorem 3.7]{LM1989},
we have $\Cl_{n}\cong\Cl_{n+1}^{\even}$, and hence
\begin{equation}
\CAR(\R^{n+1})_{\even}\cong\CAR(\R^{n}).\label{eq:CAR(Rn+1)even}
\end{equation}
By (\ref{eq:CAR(R2n)}) and (\ref{eq:CAR(Rn+1)even}), we have
\begin{equation}
\CAR(\R^{2n+1})_{\even}\cong\Mat(2^{n},\C).\label{eq:CAR(R2n+1)=00003DMat(2n)}
\end{equation}

Let $N:=\{v\in V:(v|v)=0\}$. Assume $\dim V=\infty$ and further
$\dim(V/N)=\infty$. Let $\SYM{\Odd(V)}{Odd}$ denote the the set
of (finite) odd-dimensional subspaces $W$ of $V$ where $(\cdot|\cdot)$
is positive-definite on $W$. Then by (\ref{eq:CAR(R2n+1)=00003DMat(2n)}),
$\CAR(W)$ is a finite-dimensional factor $W^{*}$-algebra (of type
I) for all $W\in\Odd(V)$.

For $W_{1},W_{2}\in\Odd(V)$, we write $W_{1}\cauindep W_{2}$ if
$W_{1}$ and $W_{2}$ are orthogonal.

Now the following is evident:
\begin{lem}
Let $\fA_{W}:=\CAR(W)_{\even}$. Then $\{\fA_{W}|W\in\Odd(V)\}$ is
a causal net of finite-dimensional factors, indexed by the directed
set $(\Odd(V),\subset)$ with the orthogonality relation $\cauindep$.
\end{lem}

Here, I considered only the odd-dimensional subspaces of $V$ so that
Working Hypothesis \ref{hyp:causal-finitist} (finitist version)%
{} holds. However, if we adopt the weakest Working Hypothesis \ref{hyp:causal-atomic}
(atomic version), it is not necessary to confine ourselves to odd
dimensional subspaces.

For $W\in\Odd(V)$, let $\Tr_{W}$ denote the canonical trace on $\fA_{W}$,
that is, if we identify $\fA_{W}$ with $\Mat(2^{n},\C)$ ($\dim W=2n+1$),
$\Tr_{W}$ is the usual trace $\Tr$ on the matrix algebra. Let
\begin{equation}
\ON_{2}(V):=\{(u,v)\in V\times V;\,\|u\|=\|v\|=1,\ (u|v)=0\}.\label{eq:def:ON2}
\end{equation}
For $(u,v)\in\ON_{2}(V)$, let
\begin{equation}
\SYM{X_{u,v}}{Xuv}:=\im\phi(u)\phi(v).\label{eq:def:Xuv}
\end{equation}
Then $X_{u,v}$ is selfadjoint and $X_{u,v}^{2}=\bbOne$. Hence
\begin{equation}
\SYM{P_{u,v}^{\pm}}{Puv+-}:=\frac{1}{2}\left(\bbOne\pm X_{u,v}\right)\label{eq:def:Puv+-}
\end{equation}
is the spectral projections of $X_{u,v}$, which is observable. Since
$P_{u,v}^{+}=P_{v,u}^{-}$, we may consider only $P_{u,v}^{+}$.

For $W\in\Odd(V)$ and $(u,v),(u',v')\in\ON_{2}(W)$, consider the
prior conditional probability
\begin{equation}
\Prob_{W}(u',v'|u,v)\equiv\Prob_{W}\left(P_{u',v'}^{+}|P_{u,v}^{+}\right):=\frac{\Tr_{W}P_{u,v}^{+}P_{u',v'}^{+}}{\Tr_{W}P_{u,v}^{+}}.\label{eq:def:P(u'v'|uv)}
\end{equation}
In fact, we can calculate the following explicit formula: 
\begin{equation}
\Prob_{W}(u',v'|u,v)=\frac{1}{2}\left(1+(u|u')(v|v')-(u|v')(v|u')\right).\label{eq:P(u'v'|uv)-law}
\end{equation}
Eq.~(\ref{eq:P(u'v'|uv)-law}) is derived from $\Tr_{W}P_{u,v}^{+}=2^{n-1}$
($\dim W=2n+1$) and the following
\begin{lem}
For any $(u,v),(u',v')\in\ON_{2}(W)$, 
\begin{equation}
\Tr_{W}X_{u,v}X_{u',v'}=2^{n}\left[(u|u')(v|v')-(u|v')(v|u')\right],\qquad(\dim W=2n+1).\label{eq:TrXuvXuv}
\end{equation}
\end{lem}

If the space $V$ (and $W\subset V$) is given a physical meaning,
Eq.~(\ref{eq:P(u'v'|uv)-law}) can be viewed as an example of a physical/empirical
law which is described without the concept of quantum state. Notice
that the r.h.s.~of (\ref{eq:P(u'v'|uv)-law}) is independent of $W$,
despite the r.h.s.~of (\ref{eq:TrXuvXuv}) being dependent on the
dimension of $W$. Generally, in our framework of causal nets, the
description of a physical law is more or less ``scope-dependent''.
However, it is desirable for a universal physical law to be less scope-dependent.

Group-theoretically, a finite-dimensional CCR algebra $\CAR(W)$ is
related to a projective unitary representation of $SO(N)$ (or a unitary
representations of $\Spin(N)$), called the \termi{spin representation}.
Since the POVM formalism of Subsection \ref{subsec:The-POVM-picture}
is suitable for the quantum systems concerning the unitary representations
of a compact Lie group, we may expect that the POVM formalism is also
suitable to describe the physical laws of the system expressed by
$\CAR(W)$. This point of view is examined in \cite{Yam2026b}.

\appendix

\section{Appendix: Proof of Proposition \ref{prop:sigma-finite-atomic}}

Although in the preceding sections, ``a $W^{*}$-algebra'' has been
referring to a $\sigma$-finite $W^{*}$-algebra, in this section
it refers to a general (i.e., possibly non-$\sigma$-finite) $W^{*}$-algebra.
Also note that when we consider a von Neumann algebra $\fR\subset\Bdd(\cH)$,
$\cH$ is not assumed to be separable.
\begin{lem}[{\cite[Chapter V, Proposition 1.1]{Tak2002}}]
\label{prop:vonNeumannComplete} If $\fR$ is a von Neumann algebra,
then the set $\Proj(\fR)$ of all projections of $\fR$ is a complete
lattice. If $\fR$ is abelian, $\Proj(\fR)$ is a complete Boolean
algebra.
\end{lem}

\begin{defn}[{e.g., \cite[Chapter 14]{GH2009}}]
 Let $\cB$ be a Boolean algebra. $q\in\cB$ is an \termi{atom}
of $\cB$ if $q\neq0$ and %
$\{p\in\cB|p\leq q\}=\{0,q\}$. The $\cB$ is said to be \termi{atomic}
if every non-zero element dominates at least one atom.
\end{defn}

\begin{lem}[{\cite[Chapter 14, Theorem 6]{GH2009}}]
\label{lem:atomicBool} Let $\cB$ be an atomic Boolean algebra,
and $X$ the set of its atoms. Let $X_{p}:=\{q\in X:q\leq p\}$. Then
the map $p\mapsto X_{p}$ is an embedding (i.e., an injective homomorphism)
of $\cB$ into $\cP(X)$ (the power set of $X$). In other words,
$\cB$ is isomorphic to $\{X_{p}|p\in\cB\}\subset\cP(X)$ by the map
$p\mapsto X_{p}$. Furthermore, $\cB$ is complete if and only if
$\{X_{p}|p\in\cB\}=\cP(X)$.
\end{lem}

\begin{lem}
\label{lem:sigma-finite-atomic-abelian}Let $\fR$ be an $\sigma$-finite
atomic abelian $W^{*}$-algebra. Then $\fR$ is isomorphic to a countable
direct sum of $\C$, i.e., $\fR\cong\bigoplus^{n}\C$ for some $n\in\N\cup\{\infty\}$.
In other words, $\fR\cong\C^{X}$, the abelian $W^{*}$-algebra of
all functions from $X$ to $\C$, for some countable set $X$.
\end{lem}

\begin{proof}
Since $\fR$ is abelian, the atomicity of $\fR$ means that the complete
Boolean algebra $\Proj(\fR)$ is atomic by Lemma \ref{prop:vonNeumannComplete}.
Let $X:=\AProj(\fR)$ be the set of the atoms of $\Proj(\fR)$. We
have $\Proj(\fR)\cong\cP(X)$ as Boolean algebras by Lemma \ref{lem:atomicBool}.
Therefore $\fR$ is isomorphic to $\C^{X}$, and furthermore since
$\fR$ is $\sigma$-finite, $X$ is countable. 
\end{proof}

\begin{proof}[Proof of Proposition \ref{prop:sigma-finite-atomic}.]
\emph{} (2)$\then$(1) is clear. We show that (1)$\then$(2). Assume
that $\fR$ is atomic. Since $\fR$ of type I%
, we find from \cite[Chapter V, Theorem 1.27]{Tak2002} that $\fR$
has the following decomposition:
\[
\fR\cong\bigoplus_{\lambda\in\Lambda}\fA_{\lambda}\,\overline{\otimes}\,\Bdd(\cK_{\lambda}),
\]
where $\Lambda$ is an index set, $\fA_{\lambda}$ are atomic abelian
von Neumann algebras, and $\cK_{\lambda}$ are Hilbert spaces. Assume
further that $\fR$ is $\sigma$-finite. Then we find that (a) $\Lambda$
is countable, (b) $\fA_{\lambda}\cong\bigoplus^{n_{\lambda}}\C$ for
some $n_{\lambda}\in\N\cup\{\infty\}$, for all $\lambda\in\Lambda$,
by Lemma \ref{lem:sigma-finite-atomic-abelian}, and (c) each $\cK_{\lambda}$
is separable. Thus we can check $\fR\cong\bigoplus_{i=1}^{n}\Bdd(\cH_{i})$
for some separable Hilbert spaces $\cH_{i}$. 
\end{proof}

\providecommand{\noopsort}[1]{}\providecommand{\singleletter}[1]{#1}%


\begin{thebibliography}{DFPS21}

\bibitem[ADC21]{AD2021}
M.~Aizenman and H.~Duminil-Copin.
\newblock Marginal triviality of the scaling limits of critical 4{D} {I}sing
  and $\phi^4_4$ models.
\newblock {\em Annals Math.}, 194(1), 2021.
\newblock arXiv:1912.07973.

\bibitem[Ara99]{Araki99}
Huzihiro Araki.
\newblock {\em Mathematical Theory of Quantum Fields}.
\newblock Oxford University Press, Oxford, 1999.
\newblock Originally published in Japanese: Ryoshiba no Suri, Iwanami Shoten,
  Tokyo, 1993; Translated by Ursula Carow-Watamura.

\bibitem[BD]{BD2015}
Marco Benini and Claudio Dappiaggi.
\newblock Models of free quantum field theories on curved backgrounds.
\newblock In R.~Brunetti, C.~Dappiaggi, K.~Fredenhagen, and J.~Yngvason,
  editors, {\em Advances in Algebraic Quantum Field Theory}, Mathematical
  Physics Studies, pages 75--124, Berlin. Springer.
\newblock arXiv:1505.04298; DOI 10.1007/978-3-319-21353-8\_3.

\bibitem[BF09]{BF2009}
C.~B\"ar and K~Fredenhagen, editors.
\newblock {\em Quantum Field Theory on Curved Spacetimes: Concepts and
  Mathematical Foundations}.
\newblock Springer, 2009.

\bibitem[BFV03]{BFV2003}
Romeo Brunetti, Klaus Fredenhagen, and Rainer Verch.
\newblock The generally covariant locality principle - a new paradigm for local
  quantum field theory.
\newblock {\em Commun. Math. Phys.}, 237:31--68, 2003.
\newblock https://doi.org/10.1007/s00220-003-0815-7, arXiv:math-ph/0112041.

\bibitem[BGP07]{BGP2007}
Christian B{\"a}r, Nicolas Ginoux, and Frank Pfäffle.
\newblock {\em Wave Equations on {L}orentzian Manifolds and Quantization}.
\newblock European Mathematical Society, Z{\"u}rich, 2007.

\bibitem[BMU14]{Bar-Mul-Udu-2014}
H.~Barnum, M.~P. M{\"u}ller, and C.~Ududec.
\newblock Higher-order interference and single-system postulates characterizing
  quantum theory.
\newblock {\em New J. Phys.}, 16:123029, 2014.
\newblock arXiv:1403.4147.

\bibitem[BR87]{BR87}
Ola Bratteli and Derek~W. Robinson.
\newblock {\em Operator Algebras and Quantum Statistical Mechanics 1, 2nd Ed.}
\newblock Springer, Berlin, 1987.

\bibitem[BR97]{BR97}
Ola Bratteli and Derek~W. Robinson.
\newblock {\em Operator Algebras and Quantum Statistical Mechanics 2, 2nd Ed.}
\newblock Springer, Berlin, 1997.

\bibitem[Cha99]{Cha1999}
Laurent Charles.
\newblock {F}eynman path integral and {T}oeplitz quantization.
\newblock {\em Helvetica Physica Acta}, 72(5/6):341--355, 1999.
\newblock
  https://www.ipht.fr\nolinebreak[1]/Docspht/articles/t98\nolinebreak[1]/093/public/publi.pdf.

\bibitem[DFPS21]{DFPS2021}
John~B. DeBrota, Christopher~A. Fuchs, Jacques~L. Pienaar, and Blake~C. Stacey.
\newblock {B}orn's rule as a quantum extension of {B}ayesian coherence.
\newblock {\em Phys. Rev. A}, 104:022207, 2021.
\newblock https://doi.org/10.1103/PhysRevA.104.022207, arXiv:2012.14397.

\bibitem[DH06]{DH2006}
Claudio D'Antoni and Stefan Hollands.
\newblock Nuclearity, local quasiequivalence and split property for {D}irac
  quantum fields in curved spacetime.
\newblock {\em Commun. Math. Phys.}, 261:133--159, 2006.
\newblock https://doi.org/10.1007/s00220-005-1398-2; arXiv:math-ph/0106028.

\bibitem[DHP09]{DHP2009}
Claudio Dappiaggi, Thomas-Paul Hack, and Nicola Pinamonti.
\newblock The extended algebra of observables for {D}irac fields and the trace
  anomaly of their stress-energy tensor.
\newblock {\em Rev. Math. Phys.}, 21:1241--1312, 2009.
\newblock https://doi.org/10.1142/S0129055X09003864; arXiv:0904.0612.

\bibitem[Dim82]{Dim1982}
J.~Dimock.
\newblock {D}irac quantum ﬁelds on a manifold.
\newblock {\em Trans. Amer. Math. Soc.}, 269:133--147, 1982.

\bibitem[DK85]{DK85}
I.~Daubechies and J.~R. Klauder.
\newblock Quantum-mechanical path integrals with {W}iener measure for all
  polynomial {Hamiltonians} {II}.
\newblock {\em J. Math. Phys.}, 26(9):2239--2256, 1985.

\bibitem[Ear11]{Ear2011}
John Earman.
\newblock The {U}nruh effect for philosophers.
\newblock {\em Studies in History and Philosophy of Modern Physics}, 42:81--97,
  2011.

\bibitem[FR16]{FR2016}
Klaus Fredenhagen and Kasia Rejzner.
\newblock Quantum field theory on curved spacetimes: axiomatic framework and
  examples.
\newblock {\em J. Math. Phys.}, 57:031101, 2016.
\newblock arXiv:1412.5125; https://doi.org/\linebreak[1]10.1063/1.4939955.

\bibitem[Fuc02]{Fuc2002}
Christopher~A. Fuchs.
\newblock Quantum mechanics as quantum information (and only a little more).
\newblock In A.~Khrenikov, editor, {\em Quantum Theory: Reconsideration of
  Foundations}. V{\"a}xjo University Press, 2002.
\newblock arXiv:quant-ph/0205039.

\bibitem[G{\'e}r19]{Ger2019}
Christian G{\'e}rard.
\newblock {\em Microlocal Analysis of Quantum Fields on Curved Spacetimes}.
\newblock European Mathematical Society, Berlin, 2019.

\bibitem[GH09]{GH2009}
Steven Givant and Paul Halmos.
\newblock {\em Introduction to Boolean Algebras}.
\newblock Springer, Berlin, 2009.

\bibitem[GJ05a]{GJ2005b}
Christian G{\'e}rard and Christian~D. J{\"a}kel.
\newblock Thermal quantum fields without cut-offs in 1+1 space-time dimensions.
\newblock {\em Rev. Math. Phys.}, 17:113--174, 2005.
\newblock arXiv:math-ph/0403048, DOI:
  https://doi.org/10.1142/S0129055X05002303.

\bibitem[GJ05b]{GJ2005}
Christian G{\'e}rard and Christian~D. J{\"a}kel.
\newblock Thermal quantum ﬁelds with spatially cutoff interactions in 1+1
  space–time dimensions.
\newblock {\em J. Funct. Anal.}, 220:157--213, 2005.
\newblock arXiv:math-ph/0307053.

\bibitem[GT09]{Tab2009}
Yousef Ghazi-Tabatabai.
\newblock {\em Quantum measure theory: a new interpretation}.
\newblock PhD thesis, Imperial College, January 2009.
\newblock arXiv:0906.0294, A Thesis Submitted for the Degree of Doctor of
  Philosophy of the University of London and the Diploma of Imperial College.

\bibitem[G{\"u}n10]{Gun10}
Batu G{\"u}neysu.
\newblock {\em On the Feynman-Kac formula for Schr\"odinger semigroups on
  vector bundles}.
\newblock PhD thesis, Rheinischen Friedrich-Wilhelms-Universit\"at Bonn, 2010.
\newblock \texttt{https://hdl.handle.net/20.500.11811/4970}.

\bibitem[G{\"u}n17]{Gun17}
Batu G{\"u}neysu.
\newblock {\em Covariant Schr{\"o}dinger Semigroups on Riemannian Manifolds}.
\newblock Operator Theory: Advances and Applications Vol. 264. Birkh{\"a}user,
  Cham, 2017.

\bibitem[Haa96]{Haa96}
Rudolf Haag.
\newblock {\em Local Quantum Physics}.
\newblock Springer, Berlin, second revised and enlarged edition, 1996.

\bibitem[Hal06]{Hal06}
Hans Halvorson.
\newblock Algebraic quantum field theory.
\newblock In Jeremy Butterfield and John Earman, editors, {\em Handbook of the
  Philosophy of Physics}, pages 731--864. Elsevier, 2006.
\newblock An appendix by Michael M\"uger,
  https://doi.org/10.1016/B978-044451560-5/50011-7, arXiv:math-ph/0602036.

\bibitem[Hor90]{Hor1990}
S.~S. Horuzhy.
\newblock {\em Introduction to algebraic quantum field theory}.
\newblock Kluwer Academic Publishers, Dordrecht, 1990.

\bibitem[HW15]{HW2015}
Stefan Hollands and Robert~M. Wald.
\newblock Quantum fields in curved spacetime.
\newblock {\em Physics Reports}, 574:1--35, 2015.
\newblock arXiv:1401.2026; https://doi.org/10.1016/j.physrep.2015.02.001.

\bibitem[KM15]{KM2015}
Igor Khavkine and Valter Moretti.
\newblock Algebraic {QFT} in curved spacetime and quasifree {H}adamard states:
  an introduction.
\newblock In R.~Brunetti et~al., editor, {\em Advances in Algebraic Quantum
  Field Theory}. Springer, 2015.
\newblock Chapter 5, arXiv:1412.5945;
  https://doi.org/10.1007/978-3-319-21353-8\_5.

\bibitem[LM89]{LM1989}
H.~Blaine Lawson, Jr and Marie-Louise Michelsohn.
\newblock {\em Spin Geometry}.
\newblock Princeton University Press, Princeton, 1989.

\bibitem[M{\"u}l21]{Mue2021}
Markus~P. M{\"u}ller.
\newblock Probabilistic theories and reconstructions of quantum theory (les
  houches 2019 lecture notes).
\newblock {\em SciPost Phys. Lect. Notes}, 28, 2021.
\newblock arXiv:2011.01286.

\bibitem[Rov96]{Rov1996}
Carlo Rovelli.
\newblock Relational quantum mechanics.
\newblock {\em Int. J. of Theor. Phys.}, 35:1637, 1996.
\newblock https://doi.org/10.1007/BF02302261, arXiv:quant-ph/9609002.

\bibitem[RS75]{RS75}
M.~Reed and B.~Simon.
\newblock {\em Fourier Analysis, Self-adjointness}.
\newblock Methods of Modern Mathematical Physics II. Academic Press, San Diego,
  1975.

\bibitem[San08]{San2008}
Ko~Sanders.
\newblock {\em Aspects of locally covariant quantum field theory}.
\newblock PhD thesis, University of York (U.K.), 2008.
\newblock arXiv:0809.4828.

\bibitem[Sor94]{Sor94}
R.~D. Sorkin.
\newblock Quantum mechanics as quantum measure theory.
\newblock {\em Mod. Phys. Lett. A}, 9:3119--3128, 1994.
\newblock arXiv:gr-qc/9401003.

\bibitem[Sum90]{Sum1990}
Stephen~J. Summers.
\newblock On the independence of local algebras in quantum ﬁeld theory.
\newblock {\em Rev. Math. Phys.}, 2:201--247, 1990.

\bibitem[SW64]{SW64}
R.~F. Streater and A.~S. Wightman.
\newblock {\em {PCT},Spin and Statistics, and All That}.
\newblock W. A. Benjamin, INC., New York, 1964.

\bibitem[Tak02]{Tak2002}
M.~Takesaki.
\newblock {\em Theory of Operator Algebras {I}}.
\newblock Springer, Berlin, 2002.
\newblock 2nd printing of the First Edition 1979.

\bibitem[Wal94]{Wal1994}
Robert~M. Wald.
\newblock {\em Quantum Field Theory in Curved Spacetime and Black Hole
  Thermodynamics}.
\newblock University of Chicago Press, Chicago, 1994.

\bibitem[Woo92]{Woo92}
N.~M.~J. Woodhouse.
\newblock {\em Geometric Quantization}.
\newblock Clarendon Press, Oxford, 2nd edition edition, 1992.

\bibitem[Yam11]{Yam11}
Hideyasu Yamashita.
\newblock Phase-space path integral and {B}rownian motion.
\newblock {\em J. Math. Phys.}, 52:022101, 2011.

\bibitem[Yam18]{Yam18}
Hideyasu Yamashita.
\newblock {G}lauber-{S}udarshan-type quantizations and their path integral
  representations for compact {L}ie groups.
\newblock \texttt{arXiv:1811.08844}, 2018.

\bibitem[Yam22a]{Yam22b}
Hideyasu Yamashita.
\newblock Antinormally-ordered quantizations, phase space path integrals and
  the {O}lshanski semigroup of a symplectic group.
\newblock \texttt{arXiv:2209.04139}, 2022.

\bibitem[Yam22b]{Yam22a}
Hideyasu Yamashita.
\newblock The {B}erezin--{S}imon quantization for {K}\"ahler manifolds and
  their path integral representations.
\newblock \texttt{arXiv:2208.12446}, 2022.

\bibitem[Yam25]{Yam2025}
Hideyasu Yamashita.
\newblock The conditional probabilities and the empirical laws in a free scalar
  {QFT} in curved spacetime.
\newblock {\em arXiv:2511.12311}, 2025.

\bibitem[Yam26a]{Yam2026b}
Hideyasu Yamashita.
\newblock The empirical laws for majorana fields in a curved spacetime.
\newblock {\em arXiv:2602.16907}, 2026.

\bibitem[Yam26b]{Yam2026a}
Hideyasu Yamashita.
\newblock A note on the conceptual problems on the unruh effect.
\newblock {\em arXiv:2602.20347}, 2026.

\bibitem[Yng05]{Yng2005}
Jakob Yngvason.
\newblock The role of type {III} factors in quantum field theory.
\newblock {\em Rept. Math. Phys.}, 55:135--147, 2005.
\newblock arXiv:math-ph/0411058.

\end{thebibliography}
\end{document}